\documentclass[11pt]{article}
\usepackage[dvipsnames,svgnames,x11names]{xcolor}
\usepackage{amsmath}
\usepackage{xspace}
\usepackage{float}
\usepackage[T1]{fontenc}
\usepackage{mdframed}
\usepackage{mathtools}
\usepackage{enumitem}
\usepackage{authblk}

\usepackage{fullpage}
\usepackage{thmtools}
\usepackage[colorlinks]{hyperref}
\hypersetup{colorlinks={true},linkcolor={blue},citecolor=magenta}

\usepackage{amssymb,amsfonts,amsmath,amsthm,amscd,dsfont,mathrsfs}
\usepackage{thmtools}
\usepackage{complexity}
\usepackage{multirow}
\usepackage{mathpazo}
\usepackage{braket}
\usepackage{comment}
\usepackage{quantikz}
\usepackage{tikz}
\usetikzlibrary{arrows.meta}
\usetikzlibrary{positioning}
\usepackage[
backend=biber,
style=alphabetic,
sorting=ynt,
maxnames=99
]{biblatex}
\usepackage{url}
\usepackage{bbm}
\usepackage[capitalize,nameinlink]{cleveref}
\usepackage{hyperref} 
\usepackage[margin=1in]{geometry}
\hypersetup{breaklinks=true}

\AddToHook{cmd/appendix/before}{\crefalias{section}{appendix}}

\theoremstyle{plain}
\newtheorem{theorem}{Theorem}[section]
\newtheorem{lemma}[theorem]{Lemma}
\newtheorem{corollary}[theorem]{Corollary}
\newtheorem{definition}[theorem]{Definition}
\newtheorem{remark}[theorem]{Remark}

\newtheorem*{theorem*}{Theorem}
\newtheorem*{definition*}{Definition}

\newmdtheoremenv[backgroundcolor=gray!10,
                 linewidth=0pt,
                 innerleftmargin=4pt,
                 innerrightmargin=4pt,
                 innertopmargin=1pt,
                 innerbottommargin=4pt,
            splitbottomskip=4pt]{problem}[prob]{Problem}

\newmdtheoremenv[backgroundcolor=gray!10,
                 linewidth=0pt,
                 innerleftmargin=4pt,
                 innerrightmargin=4pt,
                 innertopmargin=1pt,
                 innerbottommargin=5.5pt,
            splitbottomskip=4pt]{conjecture}[conj]{Conjecture}

\newcommand{\abs}[1]{\left| #1 \right|}

\mathchardef\mhyphen="2D

\newcommand{\Id}{\mathbb{1}}
\newcommand{\cols}{\{\mathcolor{red}{r},\mathcolor{green!60!black}{g}, \mathcolor{blue}{b}\}}

\newcommand{\CC}{\mathsf{C}}
\newcommand{\bfB}{\mathbf{B}}

\newcommand{\bfA}{\mathbf{A}}

\renewcommand{\proj}[1]{\ensuremath{|#1\rangle \langle #1|}}
\newcommand{\norm}[1]{\left\Vert {#1} \right\Vert}

\newcommand{\bit}{\{0,1\}}

\newcommand{\ind}{\mathds{1}}
\newcommand{\algo}{\mathcal}

\newcommand{\spn}[1]{\mathsf{span}\{#1\}}
\newcommand{\dom}{\mathsf{Dom}}
\newcommand{\im}{\mathsf{Im}}

\newcommand{\cpO}{\mathsf{cpO}}
\newcommand{\cwO}{\mathsf{cwO}}
\newcommand{\ciO}{\mathsf{ciO}}
\newcommand{\wO}{\mathsf{wO}}

\renewcommand{\cc}{\mathsf{C}}
\newcommand{\ff}{\mathsf{F}}

\newcommand{\wC}{\mathsf{wC}}
\newcommand{\wP}{\mathsf{wP}}
\newcommand{\pp}{\mathsf{P}}

\begin{document}

\title{A quantum lower bound for path finding in welded trees}
\author{Joseph Carolan}
\author{Andrew M.~Childs}
\author{Matthew Coudron}
\author{Amin Shiraz Gilani}
\affil{Department of Computer Science, Institute for Advanced Computer Studies, and \protect\\ Joint Center for Quantum Information and Computer Science, University of Maryland}
\date{}

\maketitle

\begin{abstract}
In the welded tree problem, an algorithm is tasked with navigating a graph formed from two binary trees joined at the leaves through a ``weld'' of connecting edges. A quantum walk can navigate from root to root exponentially faster than any classical algorithm. However, known efficient quantum algorithms cannot find a path between the roots, as recording the path destroys constructive interference and thus the speedup.
We prove that this is inherent: any quantum algorithm needs exponentially many queries to find a path between the roots of an independently matched welded tree graph. This provides an example of a problem that a quantum computer can solve exponentially faster than any classical algorithm by exploring exponentially many paths in superposition, but where it is provably intractable to find any such path.
The proof uses compressed permutation oracles to record the progress of a quantum algorithm as it queries the graph. We show that the compressed database remains path-free up to a small error. By controlling such errors and bounding the progress of the algorithm with each compressed oracle query, we show that $\Omega(2^{n/12})$ queries are required to find a path in a height-$n$ tree with constant success probability.
\end{abstract}

\section{Introduction}

Quantum computers can solve certain problems dramatically faster than classical ones, but the extent of this advantage remains poorly understood.
Notable families of quantum algorithms with superpolynomial speedup include ones based on finding periodic structures using Fourier transforms \cite{Shor_1997} and simulating the dynamics of quantum systems \cite{Llo96}.
However, there is both a shortage of techniques to develop qualitatively different types of quantum speedup, and an incomplete understanding of the limitations of quantum algorithms.

The \emph{welded tree problem} provides an example of a very different kind of exponential quantum speedup from those mentioned above. In this problem, one is given access to an adjacency-list oracle for a graph formed by joining the leaves of two height-$n$ balanced binary trees with a random cycle that alternates between the leaves of the two trees (the ``weld''). Given the label of the root of one of the trees (the ``entrance''), a quantum walk on this graph can find the root of the other tree (the ``exit'') in time polynomial in $n$ by traversing the graph in superposition, whereas a classical algorithm requires exponentially many queries \cite{childs03weld}.

In addition to providing a novel type of quantum speedup, the welded tree problem has yielded  insight into a variety of other questions about the power of quantum computation. For example, it was used to show that there are quantum speedups that provably cannot be achieved with only logarithmic quantum depth \cite{CoudronM19}. Furthermore, the welded tree construction has been a useful building block for other results in quantum complexity theory, on topics including the power of adiabatic optimization \cite{GV20,LWWZ25}, the quantum query complexity of graph property testing \cite{BCGKPW20}, and the possibility of exponential quantum speedups for machine learning tasks \cite{GL24}, among others.

The welded tree problem also suggests the possibility of a significant difference between the difficulty of detecting the presence of a path and actually finding such a path. Even though there have been settings where such a conditional polynomial separation is known \cite{CGP26}, an unconditional exponential separation remains elusive. While the algorithm of \cite{childs03weld} is capable of finding its way through the graph, it does this by traversing many paths in superposition and does not output any one such path. Indeed, if the algorithm were modified to store information about the path taken, this would destroy quantum interference and make the algorithm effectively classical, causing it to fail.
This phenomenon can be seen as a kind of computational analog of the double-slit experiment.
The question of whether any quantum algorithm can efficiently find a path is a longstanding open problem, alluded to in the original paper \cite{childs03weld} and more recently mentioned in a list of significant open problems in quantum query complexity \cite{Aaronson21}.
While variants of this problem have been constructed in which there remains an exponential quantum advantage for detecting a path and there \emph{is} an efficient quantum algorithm for finding a path \cite{Li23,LZ23,LT24}, the intriguing possibility of an exponential separation remains open.

Prior work made partial progress on this question by showing that there is no efficient quantum query algorithm for the welded tree path-finding problem that is both \emph{rooted} and \emph{genuine} \cite{CCG22}. Informally, a genuine algorithm is one that only provides meaningful vertex labels
as inputs to the welded tree oracle, and a rooted algorithm is one that always maintains a path from the entrance to every vertex appearing in its state. These are both reasonable conditions for an algorithm to satisfy, and this result rules out some natural candidates for path-finding algorithms, including natural instantiations of the only explicit proposal for a path-finding algorithm, the so-called snake walk analyzed by Rosmanis \cite{Rosmanis11}, in which the state of the algorithm is a polynomially long path in the graph. However, algorithms can potentially have both non-rooted and non-genuine behavior, so this work leaves the general question open.

We resolve this question (for the closely related setting of independently matched welds described below) by showing that a quantum algorithm requires exponentially many queries to find a path from the entrance to the exit of a welded tree graph. Our proof uses the compressed permutation oracle of \cite{carolan2025compressedpermutationoracles} (introduced in \cref{sec:compressed-perm}), which provides a method for showing quantum lower bounds when querying an oracle that implements a random permutation. This method has  been applied to show lower bounds for cryptographic constructions involving random permutations that resisted prior analysis \cite{cpz24precomp,carolan24oneway,MMW24,CHS19}.
Conceptually, we use the compressed permutation as a building block for our oracle, by implementing the oracle for the welded tree problem using permutations that describe a random weld.

In our analysis, we construct the weld as a union of three random matchings between the leaves of the binary trees. This differs from the weld specified in \cite{childs03weld}, which considers a random cycle. However, that choice was somewhat arbitrary, and the choice considered here also defines a problem with an exponential quantum speedup for finding the exit. It is straightforward to adapt our argument to show hardness for the original distribution of welded tree graphs, though we leave the details for future work.

\subsection*{Technical overview}

Our proof tracks what information a quantum algorithm acquires about the
weld using compressed databases for permutations that describe the weld.
To do so, we identify a family of database states that is approximately
preserved by queries. This replaces assumptions about the algorithm's behavior, as in prior work~\cite{CCG22},
with an invariant of the purified oracle. In particular, we do
not require the algorithm to retain paths to the vertices it uses, or to
query only labels it has previously obtained, instead allowing general quantum algorithms.

After formally introducing our notation for welded trees in \cref{sec:welded-tree-defs}, we construct a \emph{compressed welded tree oracle} in \cref{subsec:compressed-welded}. The randomness of
the input welded tree is specified by independent colored matchings between eligible
leaves,
as well as an independent injection assigning vertex labels from a much
larger set.
The compressed permutation oracle then represents this randomness
by a coherent database of partial permutations.
To achieve this, we extend the compressed permutation oracle technique to expanding injections in \cref{subsec:expand}: because valid labels are
sparse, an inverse labeling query can be replaced, with small error, by a
lookup in the existing database. This controls attempts to guess previously
unseen labels. Combining these ingredients
gives a faithful simulation of graph queries in terms of the compressed permutation oracle. We call this whole operator the \emph{compressed welded oracle}, $\cwO$.

\cref{sec:partial-trees} introduces the subspace that let us track the progress of an algorithm. In particular,
to describe the relevant database states, we use \emph{color words}:
sequences of edge colors specifying nonbacktracking walks from either root. Intuitively, color words capture the knowledge of efficient quantum algorithms: arbitrary polynomial-size trees of labeled vertices rooted at the left and right roots (i.e., the entrance and the exit), which never meet one another. Note that, in contrast with classical algorithms, not every vertex in these trees needs to be labeled, or in other words the algorithm can forget the explicit path it took to reach a given vertex.
A collection of words, with labels assigned to selected endpoints, describes
a weighted superposition of all minimal databases realizing those words, subject to certain constraints preventing the two sides from meeting.
We analyze the linear span of these states to establish properties of states reachable by efficient quantum algorithms.

To describe the explicit conditions preventing the left-rooted and right-rooted color words from meeting, we introduce additional terminology. The structure of the welded tree graph is organized by partitioning each height-$n$ tree into
\emph{blocks} rooted halfway down. Writing $N=2^n$, each tree has $\sqrt N$
blocks, each with $\sqrt N$ leaves. Valid realizations of color words avoid two kinds of
\emph{block violation}: a new block visit colliding with an earlier visit,
and a path leaving a block through its top after arriving through
the weld. Intuitively, a fresh weld arrival is unlikely to hit an occupied
block. Likewise, following a prescribed color sequence upward to a block
root is unlikely, since for most choices of the weld this color sequence would stay trapped near the weld.

The key lemma bounding the effect of block violations and showing preservation of the color-word subspace is the \emph{clipping lemma}, \cref{lem:clipped-subspace} in \cref{subsec:clipping-placements}. Suppose the algorithm makes a query asking for the $c$-colored neighbor of vertex $x$. For an ordinary color-word state this query may, with non-zero probability, lead to a block violation: imagine that $x$ is a block root reached through the weld, and $c$ is its neighbor furthest from the weld. The lemma evaluates the effect of \emph{clipping} color words, meaning that configurations where the queried extension leads to a violation are removed. By proving that these spaces are close, we show that the extension is unlikely to leave the color-word subspace.

The appropriate measure of information growth in this analysis is the size of the so-called \emph{core} of the database,
rather than the size of the entire compressed permutation database. In each occupied component,
the core is the minimal tree connecting recorded labels and endpoints of
recorded weld edges. It includes unlabeled connecting vertices, and labels 
outside the blocks
do not contribute
to it. The structure of this core forest is determined
by the compressed databases, so different core types give orthogonal sectors. 

However, different color-word states may share the same core and nontrivially overlap. This non-orthogonality is one of the key technical challenges for the approach. In particular, there is a degeneracy between color-word states originating from the left vertex and from the right vertex. For example, a weld edge whose path has been forgotten can be viewed as either a left-rooted or a right-rooted state. This also allows us to describe what happens to the compressed oracle database when a quantum walk algorithm traverses the graph from entrance to exit: the color words from the left root are \emph{re-rooted} to the right root.

On the other hand, we show that the color-word subspace is approximately preserved under queries, and the core size grows by $O(n)$ per query, up to a small error. While this is straightforward to show for the basis states we select,
a central difficulty is turning these small probabilities
into bounds that hold for arbitrary superpositions. We address this through a representation by
\emph{random core placements}, a tool developed in 
\cref{subsec:clipping-placements} and \cref{sec:appendix}. Within a fixed sector, cores are initially
placed independently, and a color-word state is specified by fixing an \emph{anchor},
the position of one vertex in each core. This corresponds to a component fixed by a color word before being randomized by the weld. By orthogonally decomposing the span of these anchor states, we obtain control of linear combinations.
This representation allows proving the clipping lemma.

In \cref{sec:main_lower}, we use these tools to analyze queries. First, we
clip each color-word state by retaining only databases that admit a valid extension. In \cref{lem:subspace-preservation}, we calculate the action of the decompression and compression operators inside $\cwO$. Decompression is simply an extension of the color word, so because the color-word states are clipped, this operator leads to a color-word state exactly (up to a collision-type event when sampling weld edges). Compression undoes this operation, leading back to a clipped color word, which in turn is close to the actual color-word subspace.\footnote{There is a subtlety around when compression forgets the leaf through which a component was entered: in particular, a memory in the form of non-collision with this block is maintained even after forgetting the label. We handle this through techniques developed in \Cref{subsec:clipping-placements} that control block collisions.}

Together, these estimates control both decompression and recompression,
including the creation, erasure, and resampling of terminal weld edges.
Iterating the resulting subspace-preservation bound keeps a $q$-query
algorithm close to the span of color-word states of core size $O(nq)$, as we establish in \Cref{lem:reachable-colorword}. To complete
the lower bound, we use the compressed oracle fundamental lemma to show that recovering an explicit entrance-to-exit path requires such a path to appear in the database, which in turn is overwhelmingly unlikely in the color-word subspace. The errors are
polynomial in the query count and core size and inverse exponentially small in
the tree height, giving the exponential query lower bound.

\begin{theorem}[Informal version of \Cref{thm:path-finding}]
    Any quantum algorithm that outputs an entrance-to-exit path with bounded error in the height-$n$ welded tree problem requires $\Omega(2^{n/12})$ queries.
    \label{thm:main-informal}
\end{theorem}

\paragraph{Note on concurrent work.} In the final stages of this work, we became aware of concurrent independent work that establishes a similar result \cite{MP26}.

\section{Preliminaries}

\subsection{Compressed permutation oracles}
\label{sec:compressed-perm}

The compressed permutation oracle method represents an
algorithm's interaction with a random permutation using a database
register whose computational basis states are partial permutations.
We describe the construction for a permutation on $[M]$, using
the arbitrary-size bounds of \cite{cm26comp}, for every integer $M\ge3$.

Let $\mathbf I$ denote the set of partial permutations
$I:[M]\rightarrow[M]\cup\{\bot\}$, which are injective on their
non-$\bot$ entries. Write $|I|=|\dom(I)|$, and let $\emptyset$
denote the empty database. For an input $x$, define
$\mathbf I|^x\coloneqq\{I\in\mathbf I:x\notin\dom(I)\}$. For each
$I\in\mathbf I|^x$, let
\begin{align}
    \ket{+_{x,I}}
    &\coloneqq
    \frac{1}{\sqrt{M-|I|}}
    \sum_{y\in[M]\setminus\im(I)}
        \ket{I[x\rightarrow y]},
\end{align}
where $I[x\rightarrow y]$ extends $I$ by the pair $(x,y)$.
The compression operator is
\begin{align}
    \cc_x
    &\coloneqq
    \Id+\sum_{I\in\mathbf I|^x}
    \left(
        \ket{+_{x,I}}\bra{I}
        +\ket{I}\bra{+_{x,I}}
        -\proj{I}
        -\proj{+_{x,I}}
    \right).
    \label{eq:perm-compression}
\end{align}
Thus $\cc_x$ swaps $\ket{I}$ with $\ket{+_{x,I}}$ on each
block, and fixes their orthogonal complement. In particular,
$\cc_x=\cc_x^\dagger=\cc_x^{-1}$.

We write $\cc$ for compression controlled on the query point.
The database lookup and flip operators are
\begin{align}
    \pp\ket{x,z}\ket{I}
    &\coloneqq
    \begin{cases}
        \ket{x,z\oplus I(x)}\ket{I},
            & x\in\dom(I),\\
        \ket{x,z}\ket{I},
            & x\notin\dom(I),
    \end{cases}
    &
    \ff\ket{I}
    &\coloneqq \ket{I^{-1}}.
\end{align}
Here the elements of $[M]$ have fixed distinct bit-string encodings.
With a direction bit $b$, the compressed permutation oracle is
\begin{align}
    \cpO
    &\coloneqq
    \sum_{b\in\bit}\proj{b}\otimes
        \ff^b\cc\pp\cc^\dagger\ff^b.
    \label{eq:compressed-permutation-oracle}
\end{align}
The database begins in $\ket{\emptyset}$ and is inaccessible to
the algorithm. Each query increases its size by at most one,
so after $q$ queries it is supported on databases of size at
most $q$.

\begin{theorem}[Soundness~{\cite{cm26comp}}]
    \label{thm:comp-perm-soundness}
    Let $\algo A$ make $q$ queries to a permutation and its
    inverse. Let $\rho_{\algo A}^{(\mathrm{perm})}$ be its final
    state when the permutation is uniformly random, and let
    $\rho_{\algo A}^{(\cpO)}$ be its final state with $\cpO$,
    after tracing out the database. Then
    \begin{align}
        \frac12
        \norm{
            \rho_{\algo A}^{(\mathrm{perm})}
            -\rho_{\algo A}^{(\cpO)}
        }_1
        &=O\left(\frac{q}{M^{1/2}}\right).
    \end{align}
\end{theorem}

We also use the fact that decompression produces a defined answer
on states reached by a bounded-query algorithm, up to a small
error. We formalize this by \emph{sanitizing} decompression:
the component that would produce an undefined answer is
discarded.

\begin{definition}[Validity and sanitized decompression]
    \label{def:valid}
    Let
    \begin{align}
        \Pi_x
        &\coloneqq
        \sum_{\substack{I\in\mathbf I\\x\in\dom(I)}}\proj{I},
        &
        \overline{\Pi}_x
        &\coloneqq
        \cc_x\Pi_x\cc_x^\dagger.
    \end{align}
    A database state $\ket{\psi}$ is \emph{valid} on input $x$
    if $\overline{\Pi}_x\ket{\psi}=\ket{\psi}$.
    The \emph{sanitized decompression operator} and its adjoint,
    the \emph{sanitized compression operator}, are
    \begin{align}
        \overline{\cc}_x^\dagger
        &\coloneqq
        \Pi_x\cc_x^\dagger
        =\cc_x^\dagger\overline{\Pi}_x,
        &
        \overline{\cc}_x
        &\coloneqq
        \cc_x\Pi_x
        =\overline{\Pi}_x\cc_x.
        \label{eq:sanitized-decompression}
    \end{align}
    Thus sanitized and ordinary decompression agree on valid
    inputs, while sanitized decompression annihilates the
    invalid component. For a state containing query registers,
    including the direction bit, the validity projector is
    \begin{align}
        \overline{\Pi}
        &\coloneqq
        \sum_{\substack{b\in\bit\\x\in[M]}}
        \proj{b,x}\otimes
        \ff^b\overline{\Pi}_x\ff^b.
    \end{align}
\end{definition}

Both $\overline{\cc}_x^\dagger$ and $\overline{\cc}_x$ are
contractions. Let $\overline{\cc}^\dagger$ and $\overline{\cc}$
denote their coherently controlled versions. Define the sanitized
compressed permutation oracle by
\begin{align}
    \overline{\cpO}
    &\coloneqq
    \sum_{b\in\bit}\proj{b}\otimes
        \ff^b\overline{\cc}\pp\overline{\cc}^\dagger\ff^b
    =\overline \Pi \cdot \cpO \cdot \overline{\Pi},
    \label{eq:sanitized-permutation-oracle}
\end{align}
which is also a contraction. States in the sanitized experiment
are left subnormalized, with no renormalization after
discarding an invalid component. We show that sanitizing does not significantly change the experiment, which follows from the validity of the compressed permutation oracle. 

\begin{lemma}[Validity through sanitization]
    \label{lem:valid-states}
    Let $\ket{\psi_q}$ be the joint algorithm--database state
    after $q$ queries to $\cpO$, starting with an empty database.
    Let $\ket{\overline{\psi}_q}$ be the vector obtained by
    replacing each query with $\overline{\cpO}$, keeping all
    other operations unchanged. Then
    \begin{align}
        \norm{\ket{\psi_q}-\ket{\overline{\psi}_q}}
        &=O\left(\frac{q^2}{M^{1/2}}\right).
    \end{align}
\end{lemma}

\begin{proof}
    We first establish the approximate validity estimate
    needed for the hybrid argument. If $\ket{\psi_t}$ is a
    state reached after $t$ ordinary queries, with the next
    query prepared, then
    \begin{align}
        \norm{
            (\Id-\overline{\Pi})\ket{\psi_t}
        }
        &=O\left(\frac{t+1}{M^{1/2}}\right).
        \label{eq:reachable-permutation-validity}
    \end{align}
    We use the stronger joint-state comparison from the proof of the
    soundness theorem in \cite{cm26comp}. For $r$ swap queries and
    $3r-1\le M$, the final compressed state is within Euclidean distance
    \[
        \frac{60r}{\sqrt{M-r+1}}
        +\frac{r(r+1)}{2(M-2r+1)}
    \]
    of the image of the corresponding purified state under a fixed
    isometry acting only on the permutation memory. A swap query
    interchanges a separate blank answer symbol $\bot$ with the
    oracle value; undefined database entries act as the identity.
    Each XOR query is implemented by computing into a blank auxiliary
    register, copying the encoded value, and uncomputing, using two
    swap queries. The inner compressions cancel, so this also gives
    exactly the XOR compressed oracle defined above.
    After these $2t$ swap queries, append one swap query with a fresh
    blank answer register. In the ideal experiment its answer is
    never blank. In the compressed experiment the norm of its blank
    component is exactly $\norm{(\Id-\overline\Pi)\ket{\psi_t}}$:
    decompression leaves an undefined entry precisely on this
    component, and recompression preserves its norm. The isometry
    does not act on the answer register, so the joint-state comparison
    with $r=2t+1$ bounds this norm by $O((t+1)/\sqrt M)$ when
    $t+1\le\sqrt M/4$. Outside this range the same bound is trivial.
    Since $\overline{\cpO}=\overline{\Pi} \cdot \cpO\cdot \overline{\Pi}$, we have
    \begin{align}
        \cpO-\overline{\cpO}
        &=
        (\Id-\overline{\Pi})\cpO
        +\overline{\Pi} \cdot \cpO(\Id-\overline{\Pi}).
    \end{align}
    Hence, by the triangle inequality and contractivity,
    \begin{align}
        \norm{(\cpO-\overline{\cpO})\ket{\psi_t}}
        &\leq
        \norm{(\Id-\overline{\Pi})\cpO\ket{\psi_t}}
        +\norm{(\Id-\overline{\Pi})\ket{\psi_t}}\\
        &=O\left(\frac{t+2}{M^{1/2}}\right).
    \end{align}
    The last line applies
    \eqref{eq:reachable-permutation-validity} to
    $\ket{\psi_t}$ and $\cpO\ket{\psi_t}$, which are states
    reached in the ordinary experiment after $t$ and $t+1$
    queries, respectively.

    Hybrid over the queries, using ordinary queries before
    the replaced call and sanitized queries afterwards.
    Every operation following the replacement is a
    contraction, so
    \begin{align}
        \norm{\ket{\psi_q}-\ket{\overline{\psi}_q}}
        &\leq
        \sum_{t=0}^{q-1}
        \norm{(\cpO-\overline{\cpO})\ket{\psi_t}}\\
        &\leq
        \sum_{t=0}^{q-1}
        \left(
            \norm{(\Id-\overline{\Pi})\cpO\ket{\psi_t}}
            +\norm{(\Id-\overline{\Pi})\ket{\psi_t}}
        \right)\\
        &=O\left(\frac{q^2}{M^{1/2}}\right).
    \end{align}
\end{proof}

An additional useful lemma states that removing a small number of terms from compression or decompression does not significantly change the operator.

\begin{lemma}[\cite{carolan2025compressedpermutationoracles}, Lemma~2.8]
\label{lem:restricted-compression}
Fix $x\in[M]$. For each $I\in\mathbf I|^x$, choose a nonempty
set $S_{x,I}\subseteq[M]\setminus\im(I)$, and define
\[
    \ket{+^S_{x,I}}
    \coloneqq
    \frac{1}{\sqrt{|S_{x,I}|}}
    \sum_{y\in S_{x,I}}\ket{I[x\rightarrow y]}.
\]
Let $\cc_x^S$ be obtained from \eqref{eq:perm-compression}
by replacing $\ket{+_{x,I}}$ with $\ket{+^S_{x,I}}$.
Thus $\cc_x^S$ swaps $\ket{I}$ with $\ket{+^S_{x,I}}$
and fixes their orthogonal complement.

Suppose $t+r<M$ and, whenever $|I|\le t$,
\[
    \bigl|([M]\setminus\im(I))\setminus S_{x,I}\bigr|
    \le r.
\]
Then every state $\ket{\psi}$ supported on databases of size
at most $t$, possibly entangled with ancillary registers, satisfies
\[
    \norm{(\cc_x-\cc_x^S)\ket{\psi}}
    \le 2\sqrt{\frac{r}{M-t}}\,\norm{\ket{\psi}}.
\]
The same bound holds for coherently controlled operators and
for sanitized compression and decompression, defined using
the same projector $\Pi_x$.
\end{lemma}

\subsection{Compressed expanding oracles}
\label{subsec:expand}
An expanding injection has a useful additional property: an inverse
query is unlikely to discover an image point that is not already in
the compressed database. We formalize this by replacing inverse
queries with direct database lookup. We also sanitize forward
decompression, discarding any component on which it fails to expose
a defined answer. The resulting oracle does nothing on an inverse
query whose image point is absent from the database.

Let $N<M$, and let $f:[N]\hookrightarrow[M]$ be a uniformly random
injection. We use fixed encodings of the finite sets in query and
answer registers, reserving the all-zero answer for $\bot$. The
compressed permutation construction and its soundness proof apply
with these fixed encodings as well. Our convention that a query
returns nothing means that it acts as the identity on the answer
register. Thus
\begin{align}
    \algo O_f\ket{b,x,z}
    &\coloneqq
    \begin{cases}
        \ket{b,x,z\oplus f(x)},
            & b=0,\ x\in[N],\\
        \ket{b,x,z\oplus f^{-1}(x)},
            & b=1,\ x\in\im(f),\\
        \ket{b,x,z}, & \text{otherwise}.
    \end{cases}
\end{align}
Identify $[N]$ with a fixed subset of $[M]$. Restricting a uniform
permutation $\pi$ of $[M]$ to this subset gives a uniform $f$:
each injection has exactly $(M-N)!$ permutation extensions.

We first construct an oracle whose database $I$ is a partial
permutation of $[M]$. Define two database lookup operators by
\begin{align}
    \pp_0\ket{x,z}\ket{I}
    &\coloneqq
    \begin{cases}
        \ket{x,z\oplus I(x)}\ket{I},
            & x\in[N]\cap\dom(I),\\
        \ket{x,z}\ket{I}, & \text{otherwise},
    \end{cases}\\
    \pp_1\ket{x,z}\ket{I}
    &\coloneqq
    \begin{cases}
        \ket{x,z\oplus I^{-1}(x)}\ket{I},
            & x\in\im(I),\ I^{-1}(x)\in[N],\\
        \ket{x,z}\ket{I}, & \text{otherwise}.
    \end{cases}
\end{align}
Write $\cc_0=\cc$ and $\cc_1=\ff\cc\ff$, with compression
controlled on $x$ as before. These are the compression operators
for the extending permutation on $[M]$. Define
\begin{align}
    A&\coloneqq\sum_{x\in[N]}\proj{x},
    &\Pi_0&\coloneqq
        \sum_{x\in[M]}\proj{x}\otimes
        \sum_{I:\,x\in\dom(I)}\proj{I},\\
    \overline{\Pi}_0&\coloneqq\cc_0\Pi_0\cc_0^\dagger,\\
    \overline{\cc}_0^\dagger
        &\coloneqq\Pi_0\cc_0^\dagger,
    &\overline{\cc}_0&\coloneqq\cc_0\Pi_0.
\end{align}
As usual, operators act as the identity on unmentioned registers.

\begin{definition}
    The \emph{compressed injective oracle} and its
    \emph{sanitized} version are
    \begin{align}
        \ciO
        &\coloneqq
        \sum_{b\in\bit}\proj{b}\otimes
            \cc_b\pp_b\cc_b^\dagger,\\
        \overline{\ciO}
        &\coloneqq
        \proj{0}\otimes
        \left(
            \overline{\cc}_0\pp_0
            \overline{\cc}_0^\dagger A+A^\perp
        \right)
        +\proj{1}\otimes\pp_1.
        \label{eq:inverse-sanitized-injection}
    \end{align}
    The database is initialized to $\ket{\emptyset}$.
    \label{def:comp-inj}
\end{definition}

The oracle $\ciO$ is unitary, whereas $\overline{\ciO}$ is a
contraction. Sanitized states are left subnormalized, without
renormalization. The term $A^\perp$ ensures that forward queries
outside $[N]$ remain identity operations. Since $\pp_0$ preserves
the database, it commutes with $\Pi_0$, and hence
\begin{align}
    \overline{\cc}_0\pp_0\overline{\cc}_0^\dagger
    &=\cc_0\pp_0\overline{\cc}_0^\dagger
      =\cc_0\pp_0\cc_0^\dagger\overline{\Pi}_0.
    \label{eq:forward-injection-sanitization}
\end{align}
Thus sanitizing both compression and decompression has the same
effect here as sanitizing decompression alone. Inverse queries to
$\overline{\ciO}$ do not change the database. Starting from the
empty database, it remains supported on partial injections with
domain contained in $[N]$, and each forward query increases
database size by at most one.

For the proof, introduce the auxiliary oracle which modifies only
inverse queries:
\begin{align}
    \ciO^{\mathsf{db}}
    &\coloneqq
    \proj{0}\otimes\cc_0\pp_0\cc_0^\dagger
    +\proj{1}\otimes\pp_1.
\end{align}
This auxiliary oracle is unitary. Define the projector onto the
invalid part of an active forward query by
\begin{align}
    R_{\mathsf{bad}}
    &\coloneqq
    \proj{0}\otimes A(\Id-\overline{\Pi}_0).
\end{align}
Equation~\eqref{eq:forward-injection-sanitization} gives
\begin{align}
    \overline{\ciO}
    &=\ciO^{\mathsf{db}}(\Id-R_{\mathsf{bad}}).
    \label{eq:injection-validity-factorization}
\end{align}

For a database operator $T$, write
$\norm{T}_{\leq t}\coloneqq\norm{T\Pi_{\leq t}^{\mathsf{db}}}$,
where $\Pi_{\leq t}^{\mathsf{db}}$ projects onto databases of
size at most $t$.

\begin{lemma}[Inverse sanitization]
    For $0\leq t<M$,
    \begin{align}
        \norm{\ciO-\ciO^{\mathsf{db}}}_{\leq t}
        &\leq 4\sqrt{\frac{2N}{M-t}}.
        \label{eq:inverse-sanitization-bound}
    \end{align}
    In particular, the bound is $O(\sqrt{N/M})$ when $t\leq M/2$.
    \label{lem:inverse-sanitization}
\end{lemma}

\begin{proof}
    The forward branches agree. For the inverse branch, flip the
    database and fix the query point $x$. On a compression block
    indexed by $J\in\mathbf I|^x$, write $\ket{e}=\ket{J}$ and
    $\ket{s}=\ket{+_{x,J}}$. On this block,
    \begin{align}
        \cc_x-\Id
        &=-(\ket{e}-\ket{s})(\bra{e}-\bra{s}).
    \end{align}
    Let $R$ be the inverse lookup $\pp_1$ in this flipped picture,
    and let $Q$ project onto extensions $J[x\rightarrow u]$ with
    $u\in[N]$. Then
    \begin{align}
        R-\Id&=Q(R-\Id)Q,
        &\norm{R-\Id}&\leq 2,\\
        Q\ket{e}&=0,
        &\norm{Q\ket{s}}^2
        &=\frac{|[N]\setminus\im(J)|}{M-|J|}
        \leq\frac{N}{M-|J|}.
    \end{align}
    It follows that
    $\norm{(\cc_x-\Id)Q}\leq\sqrt{2N/(M-|J|)}$.
    Using $\cc_x^2=\Id$, we obtain
    \begin{align}
        \cc_xR\cc_x-R
        &=(\cc_x-\Id)(R-\Id)\cc_x
          +(R-\Id)(\cc_x-\Id),\\
        \norm{\cc_xR\cc_x-R}
        &\leq 4\sqrt{\frac{2N}{M-|J|}}.
    \end{align}
    Every block meeting the size-$t$ subspace has $|J|\leq t$.
    Taking the maximum over these orthogonal blocks, and then
    over the coherently controlled query point, proves the claim.
\end{proof}

Write $\varepsilon_{\mathsf{perm}}(q,L)$ for a uniform
trace-distance soundness bound for $q$ queries to a compressed
permutation on a set of size $L$. Write $\nu_{\mathsf{perm}}(r,L)$
for a nondecreasing bound on
$\norm{(\Id-\overline{\Pi})\ket{\psi}}$ for normalized states
reached after at most $r$ ordinary permutation queries, with the
next query prepared. Define
\begin{align}
    \sigma_{\mathsf{perm}}(Q,L)
    &\coloneqq 2\sum_{r=1}^{Q}\nu_{\mathsf{perm}}(r,L).
\end{align}
The hybrid proof of \Cref{lem:valid-states} shows that this bounds
the joint-state error from sanitizing any chosen subset of at most
$Q$ permutation calls. Indeed, a replacement after $r$ ordinary
calls costs at most
$\nu_{\mathsf{perm}}(r,L)+\nu_{\mathsf{perm}}(r+1,L)$.
These bounds allow arbitrary workspace operations between calls,
including queries to independent component oracles.

We keep the parameters explicit when component permutations have
different sizes. By \Cref{thm:comp-perm-soundness} and the reachable
validity estimate in the proof of \Cref{lem:valid-states},
\begin{align}
    \varepsilon_{\mathsf{perm}}(q,L)
    &=O(q/L^{1/2}),\\
    \nu_{\mathsf{perm}}(r,L)
    &=O((r+1)/L^{1/2}),\\
    \sigma_{\mathsf{perm}}(Q,L)
    &=O(Q^2/L^{1/2}).
\end{align}

\begin{lemma}[Forward sanitization]
    For every joint state vector $\ket{\psi}$,
    \begin{align}
        \norm{(\ciO^{\mathsf{db}}-\overline{\ciO})\ket{\psi}}
        &=\norm{R_{\mathsf{bad}}\ket{\psi}}.
        \label{eq:forward-sanitization-error}
    \end{align}
    Consequently, if $\ket{\psi}$ is a normalized state reached
    by an ordinary compressed-permutation circuit using at most
    $r$ calls to the extending permutation, and its database has
    size at most $t<M$, then
    \begin{align}
        \norm{(\ciO-\overline{\ciO})\ket{\psi}}
        &\leq\nu_{\mathsf{perm}}(r,M)
          +4\sqrt{\frac{2N}{M-t}}.
        \label{eq:injection-sanitization-error}
    \end{align}
    In particular, after $t$ ordinary injection queries we may
    take $r=2t$.
    \label{lem:forward-sanitization}
\end{lemma}

\begin{proof}
    Equation~\eqref{eq:injection-validity-factorization} and the
    unitarity of $\ciO^{\mathsf{db}}$ give the first equality.
    The projector $R_{\mathsf{bad}}$ restricts to active forward
    queries, so the reachable validity bound gives
    $\norm{R_{\mathsf{bad}}\ket{\psi}}
      \leq\nu_{\mathsf{perm}}(r,M)$.
    The triangle inequality and \Cref{lem:inverse-sanitization}
    now give the second bound. Finally, each ordinary injection
    query is implemented using at most two permutation queries,
    as shown below.
\end{proof}

The forward estimate applies to reachable states. Unlike inverse
sanitization, it need not be small on every vector supported on
small databases, since forward sanitization annihilates invalid
components.

\begin{corollary}[Soundness of expanding oracles]\label{cor:expanding-soundness}
    Let $\algo A$ make $q<M$ queries. Let
    $\rho_{\algo A}^{(f)}$, $\rho_{\algo A}^{(\ciO)}$, and
    $\rho_{\algo A}^{(\overline{\ciO})}$ denote its final reduced
    states in the three experiments. The last state is
    subnormalized. Then
    \begin{align}
        \frac12\norm{
            \rho_{\algo A}^{(f)}-\rho_{\algo A}^{(\ciO)}
        }_1
        &\leq\varepsilon_{\mathsf{perm}}(2q,M),\\
        \frac12\norm{
            \rho_{\algo A}^{(f)}
            -\rho_{\algo A}^{(\overline{\ciO})}
        }_1
        &\leq\varepsilon_{\mathsf{perm}}(2q,M)
          +\sigma_{\mathsf{perm}}(2q,M)
          +4q\sqrt{\frac{2N}{M-q}}.
    \end{align}
\end{corollary}

\begin{proof}
    An injection query can be implemented with two permutation
    queries: compute $\pi(x)$ or $\pi^{-1}(x)$ into a clean
    register, add the result to the answer if the appropriate
    domain condition holds, and uncompute. Replace these
    permutation calls by compressed permutation calls.
    The two internal compression operators commute with the
    conditional copy and cancel. The remaining operation is
    precisely $\cc_b\pp_b\cc_b^\dagger$, and the temporary
    register returns to zero. Permutation soundness therefore
    gives the first inequality.

    For the second, let $\ket{\psi_t}$ be the state immediately
    before call $t+1$ in the ordinary compressed injection
    experiment. Its database has size at most $t$, and it is
    reached using at most $2t$ component permutation calls.
    Hybrid over the injection calls, using ordinary calls before
    each replacement and sanitized calls afterwards. Every
    suffix is a contraction. Hence, if $\ket{\psi_q}$ and
    $\ket{\overline{\psi}_q}$ are the final joint vectors,
    \begin{align}
        \norm{\ket{\psi_q}-\ket{\overline{\psi}_q}}
        &\leq\sum_{t=0}^{q-1}
            \norm{(\ciO-\overline{\ciO})\ket{\psi_t}}\\
        &\leq\sum_{t=0}^{q-1}
            \left(
                \nu_{\mathsf{perm}}(2t,M)
                +4\sqrt{\frac{2N}{M-t}}
            \right)\\
        &\leq\sigma_{\mathsf{perm}}(2q,M)
            +4q\sqrt{\frac{2N}{M-q}},
    \end{align}
    by \Cref{lem:forward-sanitization} and monotonicity of
    $\nu_{\mathsf{perm}}$.

    To pass to reduced states, for any vectors $\ket{a}$ and
    $\ket{b}$ of norm at most one,
    \begin{align}
        \frac12\norm{
            \ket{a}\bra{a}-\ket{b}\bra{b}
        }_1
        &\leq
        \frac{\norm{\ket{a}}+\norm{\ket{b}}}{2}
            \norm{\ket{a}-\ket{b}}\\
        &\leq\norm{\ket{a}-\ket{b}}.
        \label{eq:subnormalized-vector-distance}
    \end{align}
    Apply this inequality, trace out the database, and combine
    with the first soundness bound by the triangle inequality.
\end{proof}

In particular, when $q\leq M/2$ and the quantitative permutation
bounds apply, the sanitized expanding-oracle error is
\begin{align}
    O\left(\frac{q^2}{M^{1/2}}+q\sqrt{\frac{N}{M}}\right).
\end{align}

\section{Welded trees}
\label{sec:welded-trees}

A welded tree multigraph is defined by a vertex set $V$, edges $E$, and a coloring $C$. We use the two binary trees of \cite{childs03weld}, joined by the independent colored matchings below in place of a single cycle, with a proper three-edge-coloring. Let $\mathbf T$ be the set of such multigraphs.

\subsection{Constructing welded trees}
\label{sec:welded-tree-defs}

We construct a $3$-colored welded tree in correspondence with a \emph{labeling} injection $P$ and \emph{weld} matchings $W_r, W_b, W_g$. We begin with two trees of depth $n$, where we label such vertices by $V=\{\underline a_b, \overline a_b \, : \, b \in \{0,\dots,n\}, a \in \{0, \dots, 2^{b}-1\}\}$, and add edges for every $b \in \{0, \dots, n-1\}$, $a \in \{0, \dots, 2^b-1\}$, add edges $\underline E_T = \{\underline a_b, \underline {(2a)}_{b+1}\}, \{\underline a_b, \underline{(2a+1)}_{b+1}\}$, and similarly $\overline E_T = \{\overline{a}_b, \overline{(2a)}_{b+1}\}, \{\overline{a}_b, \overline{(2a+1)}_{b+1}\}$.

We refer to the underlined vertices as left vertices, and the overlined vertices as right vertices, imagining the tree laid out left-to-right. In this picture, internal vertices have a top and a bottom child. The set ${\mathbf L} \coloneqq \underline {\mathbf L} \sqcup \overline {\mathbf L}$ denotes the set of leaves, with $\underline {\mathbf L} \coloneqq \{\underline a_n \, : \, a \in \{0, \dots, 2^{n}-1\}\}$ the left vertices and $\overline {\mathbf L} \coloneqq \{\overline a_n \, : \, a \in \{0, \dots, 2^n-1\}\}$ the right leaf vertices. We $3$-color these trees by taking, say, blue as the excluded color at the roots $\underline 0_0, \overline 0_0$, which fixes the coloring up to the symmetries of the tree. For specificity we place the lesser color as the lower child under the ordering $r \prec b \prec g$. For a given leaf vertex $v \in \mathbf L$, we say that the \emph{color} of $v$ is the color of the edge incident to $v$. Let $\mathbf L_c \subset \mathbf L$ denote the subset of leaves of a given color $c$, and similarly for $\underline{\mathbf L}_c$ on the left and $\overline{\mathbf L}_c$ on the right.

We take $W_r : \underline{\mathbf L}_b \sqcup \underline{\mathbf L}_g \hookrightarrow \overline{\mathbf L}_b \sqcup \overline{\mathbf L}_g$ to be a bijection between the non-red left leaf vertices and the non-red right leaf vertices, and similarly for $W_b, W_g$. These functions define colored weld edges $E_W \ \{(u, v, c) \, : \, c\in\cols,\ W_c(u)=v \}$. Parallel edges of different colors are allowed. We take $P : V \hookrightarrow \bit^{2n}$ as an injection from the set of vertices to a much larger label set.

\begin{definition}
    For $P, W$ defined above, we identify a corresponding welded tree $T_{P, W} \in \mathbf T$ defined by vertices $P(V)$ and edges obtained by applying $P$ to the endpoints of the fixed colored tree edges and the colored weld edges $E_W$. We let $\algo D$ denote the distribution of these multigraphs when $P$ is a uniform injection and the $W_c$ are independent uniform bijections on their specified sets, independent also of $P$.
    \label{def:weld-trees-def}
\end{definition}

\begin{remark}
    Sampling the $P$ and $W_c$ as uniformly random injections makes $T_{P, W}$ potentially a multi-graph.
    \label{rem:unif-W}
\end{remark}

\subsection{Compressed welded tree oracle}
\label{subsec:compressed-welded}

A query specifies a vertex label $x$ and an edge color $c$.
To answer it, we decode $x$ to a vertex, follow its $c$-colored
edge, and return the label of the resulting vertex. The edges
inside the two trees are fixed; only crossing the weld requires
a query to one of the $W_c$. In the compressed picture, the
expansion of the label function lets us decode by inspecting
the database directly. If $x$ has not been recorded, the
compressed oracle does nothing.

We use the ensemble in which $P$ is a uniform injection and the
three $W_c$ are independent uniform bijections, independent also
of $P$. Thus different colors may give parallel weld edges, as
in \Cref{rem:unif-W}. Write
\begin{align}
    K&\coloneqq |V|,
    &M&\coloneqq |\bit^{2n}|,
    &K_c&\coloneqq
        |\underline{\mathbf L}_{\neg c}|
        =|\overline{\mathbf L}_{\neg c}|,
\end{align}
where $\mathbf L_{\neg c}=\bigcup_{d\neq c}\mathbf L_d$,
with the analogous convention on each side.

For either complete or partial databases, define the neighbor
function
\begin{align}
    \mathcal N_c^W(u)
    &\coloneqq
    \begin{cases}
        v, & \{u,v\}\text{ is a fixed tree edge of color }c,\\
        W_c(u), & u\in\underline{\mathbf L}_{\neg c},\\
        W_c^{-1}(u), & u\in\overline{\mathbf L}_{\neg c},\\
        \bot, & \text{otherwise}.
    \end{cases}
\end{align}
An undefined database lookup returns $\bot$, and we set
$P(\bot)=\mathcal N_c^W(\bot)=\bot$. In particular, a root
has no blue neighbor. Define
\begin{align}
    r_{P,W}(x,c)
    &\coloneqq P\bigl(\mathcal N_c^W(P^{-1}(x))\bigr).
\end{align}
For complete $P,W$, the ordinary welded tree oracle returns
$r_{P,W}(x,c)$. For partial $P,W$, the corresponding database
lookup operator is
\begin{align}
    \wP\ket{x,c,z}\ket{P,W}
    &\coloneqq
    \begin{cases}
        \ket{x,c,z\oplus r_{P,W}(x,c)}\ket{P,W},
            &r_{P,W}(x,c)\neq\bot,\\
        \ket{x,c,z}\ket{P,W}, &\text{otherwise}.
    \end{cases}
    \label{eq:weld-database-lookup}
\end{align}
Thus a nonexistent input vertex, a missing edge, or an undefined
required entry gives no answer.

The compressed oracle exposes the needed weld edge and neighbor
label before applying this lookup, and recompresses afterwards.
We use superscripts to indicate the database on which a local
compression acts; all unmentioned registers are left unchanged.
For a fixed $x,c$, put $u=P^{-1}(x)$. Define edge compression by
\begin{align}
    \CC_E\ket{x,c}\ket{P,W}
    &\coloneqq
    \ket{x,c}\ket{P}\otimes
    \begin{cases}
        \cc_u^{W_c}\ket{W},
            &u\in\underline{\mathbf L}_{\neg c},\\
        \ff_{W_c}\cc_u^{W_c}\ff_{W_c}\ket{W},
            &u\in\overline{\mathbf L}_{\neg c},\\
        \ket{W}, &\text{otherwise}.
    \end{cases}
    \label{eq:weld-edge-compression}
\end{align}
In the second branch, compression acts on the flipped database,
so its input is a right leaf. Define label compression by
\begin{align}
    \CC_L\ket{x,c}\ket{P,W}
    &\coloneqq
    \ket{x,c}\otimes
    \begin{cases}
        (\cc_v^P\otimes\Id_W)\ket{P,W},
            &u\neq\bot,\ v=\mathcal N_c^W(u)\neq\bot,\\
        \ket{P,W}, &\text{otherwise}.
    \end{cases}
    \label{eq:weld-label-compression}
\end{align}
These controls are preserved by the operation they select.
In particular, $v\neq u$, and compressing $P(v)$ preserves the
source pair $P(u)=x$. Consequently, both $\CC_E$ and $\CC_L$
are unitary involutions.

\begin{definition}
    \label{def:compressed-welded-oracle}
    The welded compression operator and compressed welded tree
    oracle are
    \begin{align}
        \wC&\coloneqq\CC_E\CC_L,
        &\cwO&\coloneqq\wC \cdot \wP \cdot \wC^\dagger.
    \end{align}
\end{definition}

Reading from right to left, a query first decompresses the weld
edge, when needed, then the neighbor's label, applies the purified query $\wP$,
and recompresses in reverse order. If $x\notin\im(P)$, every
step is the identity. Each query increases the label database
and the selected weld database by at most one entry each.

It will be useful to sanitize both compression and decompression. Let
$\overline{\CC}_E$ and $\overline{\CC}_L$ be given by
\eqref{eq:weld-edge-compression} and
\eqref{eq:weld-label-compression}, replacing each local
$\cc$ by $\overline{\cc}$ from \Cref{def:valid}.
Identity branches remain identity. In particular, the active label
branch uses the forward compression and decompression of
$\overline{\ciO}$ from \Cref{def:comp-inj}. Set
\begin{align}
    \overline{\wC}
    &\coloneqq\overline{\CC}_E\overline{\CC}_L,
    &\overline{\cwO}
    &\coloneqq
        \overline{\wC} \cdot \wP \cdot \overline{\wC}^{\dagger}.
    \label{eq:sanitized-welded-oracle}
\end{align}
These are contractions, and sanitized states are left
subnormalized.

We next justify the simulation. If the algorithm is initially
given labels of specified vertices, generate those labels by
forward queries to $P$. In particular, we give the algorithm the label of the left vertex in our setup, though one could also consider giving both the left and right vertex, or further additional vertices. Use $\algo O_P$, $\ciO$, or
$\overline{\ciO}$ for this initialization in the ordinary,
compressed, or sanitized compressed experiment, respectively.
These queries are included in the query count below; the compressed
database starts empty before initialization, and the sanitized
experiment is not renormalized after initialization.

\begin{lemma}[Faithful simulation]
    Suppose an algorithm receives $s$ labels in this manner and
    makes $q$ welded tree queries. Let $Q=O(q+s)$ bound the
    number of component permutation calls in the evaluation
    circuit described below, including initialization, and
    assume $Q<M$. Let $\rho_{\algo A}^{(\wO)}$ and
    $\rho_{\algo A}^{(\cwO)}$ be its final reduced states
    with the ordinary and compressed welded tree oracles.
    Then
    \begin{align}
        \frac12\norm{
            \rho_{\algo A}^{(\wO)}-\rho_{\algo A}^{(\cwO)}
        }_1
        &\leq
        \varepsilon_{\mathsf{perm}}(Q,M)
        +\sum_{c\in\cols}\varepsilon_{\mathsf{perm}}(Q,K_c)
        +O\left(q\sqrt{\frac{K}{M-Q}}\right).
        \label{eq:weld-simulation-bound}
    \end{align}
    \label{lem:weld-simulation}
\end{lemma}

\begin{proof}
    Implement an ordinary welded query coherently as follows:
    compute $u=P^{-1}(x)$ into a clean register; compute
    $v=\mathcal N_c^W(u)$ into another clean register; add
    $P(v)$ to the answer when $v\neq\bot$; and reverse the
    two computations of $v$ and $u$. Undefined branches do
    nothing. This uses two inverse injection queries, one
    forward injection query, and at most two queries to the
    selected $W_c$, with the inverse direction used at right
    leaves. Each injection query uses at most two queries to
    its extending permutation. Hence $Q=O(q+s)$.

    Replace the independent component permutation oracles by
    their compressed versions, one component at a time.
    Include the other components in the workspace for each
    hybrid. Permutation soundness and
    \Cref{cor:expanding-soundness} give the first two terms
    of \eqref{eq:weld-simulation-bound}.

    Next replace the two inverse label queries in each welded
    query by the direct lookup branch of
    $\ciO^{\mathsf{db}}$, retaining ordinary forward label
    queries. By \Cref{lem:inverse-sanitization}, the cost of each replacement in joint state norm is $O(\sqrt{K/(M-Q)})$. Hybrid over these $2q$ replacements to obtain
    the last term of \eqref{eq:weld-simulation-bound}.

    It remains to identify the resulting circuit with $\cwO$.
    The source vertex $u$ is now obtained by direct database
    lookup. On a weld branch, write the two weld queries as
    compression--lookup--compression. Their inner compression
    operators commute with the intervening label query:
    the weld compression acts on $W_c$ and is controlled by
    the saved $u$, whereas the label query acts on $P$ and is
    controlled by the saved $v$. The inner compressions
    therefore cancel. The two remaining weld lookups compute
    and uncompute $v$, leaving precisely the controls in
    $\CC_L$ and the lookup $\wP$. The outer compressions are
    $\CC_E$ and $\CC_E^\dagger$. This gives
    \begin{align}
        \CC_E\CC_L \cdot \wP \cdot \CC_L^\dagger\CC_E^\dagger
        &=\wC \cdot \wP \cdot \wC^\dagger.
    \end{align}
    The fixed-tree branches give the same expression with no
    weld compression. Finally, the source register uncomputes
    exactly because the source pair is preserved. All
    temporary registers return to zero, so the reduced state
    is exactly the state obtained using $\cwO$.
\end{proof}

\begin{corollary}[Sanitized simulation]
    \label{cor:sanitized-weld-simulation}
    With the notation of \Cref{lem:weld-simulation}, let
    $\rho_{\algo A}^{(\overline{\cwO})}$ be the subnormalized
    reduced state in the sanitized experiment. Then
    \begin{align}
        \frac12\norm{
            \rho_{\algo A}^{(\wO)}
            -\rho_{\algo A}^{(\overline{\cwO})}
        }_1
        &\leq
        \varepsilon_{\mathsf{perm}}(Q,M)
        +\sum_{c\in\cols}\varepsilon_{\mathsf{perm}}(Q,K_c)
        \nonumber\\
        &\quad+
        \sigma_{\mathsf{perm}}(Q,M)
        +\sum_{c\in\cols}\sigma_{\mathsf{perm}}(Q,K_c)
        +O\left(q\sqrt{\frac{K}{M-Q}}\right).
    \end{align}
\end{corollary}

\begin{proof}
    Start with the ordinary componentwise compressed evaluation
    circuit, before making any direct database substitutions.
    Replace every label injection call by $\overline{\ciO}$,
    including the $s$ initialization calls, and every weld
    permutation call by $\overline{\cpO}$. Take hybrids with
    ordinary calls before the replaced call and modified calls
    afterwards. All suffixes are contractions, and all prefixes
    are ordinary compressed-permutation circuits.

    For the label calls, \Cref{lem:inverse-sanitization,lem:forward-sanitization}
    separate the two costs. There are $2q$ inverse label calls;
    replacing them by direct lookup costs
    $O(q\sqrt{K/(M-Q)})$ in total. At an active forward label
    call preceded by $r$ label-permutation calls, the loss is
    at most $\nu_{\mathsf{perm}}(r,M)$, by
    \eqref{eq:forward-sanitization-error} and reachable validity.
    Summing over forward calls, including initialization, gives
    at most $\sigma_{\mathsf{perm}}(Q,M)$. These estimates are
    applied to the ordinary prefixes, not to states that already
    contain direct database substitutions.

    Likewise, sanitizing the calls to each $W_c$ contributes at
    most $\sigma_{\mathsf{perm}}(Q,K_c)$. Thus the total joint
    vector distance from the ordinary compressed evaluation
    circuit to the modified circuit is at most
    \begin{align}
        \sigma_{\mathsf{perm}}(Q,M)
        +\sum_{c\in\cols}\sigma_{\mathsf{perm}}(Q,K_c)
        +O\left(q\sqrt{\frac{K}{M-Q}}\right).
    \end{align}

    The circuit again reduces to the definition of
    $\overline{\cwO}$ in \eqref{eq:sanitized-welded-oracle}. To see this for a weld branch, write
    a sanitized weld query as
    $\cc\Pi\pp\Pi\cc^\dagger$, where $\Pi$ projects onto
    a defined weld entry in the appropriate direction.
    Between the two weld queries, the inner ordinary
    compressions cancel as before. The remaining copies of
    $\Pi$ commute with the intervening label operation and
    the weld lookup, and combine by $\Pi^2=\Pi$.
    The forward label call is now exactly
    $\overline{\cc}_v^P\pp_0^P
      (\overline{\cc}_v^P)^\dagger$ on its active branch,
    and is the identity when $v=\bot$. What remains is therefore
    \begin{align}
        \overline{\CC}_E\overline{\CC}_L
        \cdot \wP \cdot \overline{\CC}_L^\dagger\overline{\CC}_E^\dagger
        &=\overline{\cwO}.
    \end{align}
    The fixed-tree and inactive branches reduce in the same
    way. Sanitized label operations still preserve the source
    pair, so the temporary registers return to zero on the
    surviving, subnormalized vector. Initialization also agrees
    by its definition above.

    Finally, use \eqref{eq:subnormalized-vector-distance}, trace
    out the databases, and add the component permutation
    soundness bounds by the triangle inequality. This proves
    the stated bound without renormalizing any sanitized state.
\end{proof}

Thus the simulation error consists of the component permutation
errors and the cost of guessing a fresh label. In particular,
whenever the quantitative permutation bounds above apply and
$Q\leq M/2$, the respective errors are
\begin{align}
    O\left(
        Q\left(M^{-1/2}+\sum_{c\in\cols}K_c^{-1/2}\right)
        +q\sqrt{K/M}
    \right),\\
    O\left(
        Q^2\left(M^{-1/2}+\sum_{c\in\cols}K_c^{-1/2}\right)
        +q\sqrt{K/M}
    \right).
\end{align}
The bounds in terms of $\varepsilon_{\mathsf{perm}}$ and
$\sigma_{\mathsf{perm}}$ keep explicit which permutation
soundness estimates are needed for the actual weld sizes $K_c$.

\subsection{Partial trees and color words}
\label{sec:partial-trees}

We can formally define a partial welded tree as follows.

\begin{definition}
    A \emph{partial welded tree} is specified by a tuple $(E, C, \ell)$, where $E \subset \underline{\mathbf L} \times \overline{\mathbf L} \times \cols$ is a set of colored weld edges, $C(u,v,c)=c$ is their edge coloring, and $\ell : V \rightarrow \bit^{2n} \cup \{\bot\}$ is a partial label function which is injective on non-$\bot$ symbols.
    \label{def:partial-welded}
\end{definition}

We take $T_{P, W}$ for $P$ and each $W_c$ as partial injections to represent a \emph{partial} welded tree, as in \Cref{def:partial-welded}, in the natural way. In particular, we define $\algo T_{P, W} = (E', C', \ell)$ as \begin{align*}
    E' &\coloneqq \{(u, v, c) \, : \, c \in\cols,\ W_c(u)=v\}, \\
    C'(u, v, c) &\coloneqq c, \\
    \ell &\coloneqq P.
\end{align*}

Intuitively, one should think of this tree as containing all of the edges not in the weld; these edges appear regardless of the weld and label values. With this in place, we can define color words and color-word states.

\begin{definition}
    \label{def:color-word}
    A \emph{color word} of length $t$ is a (potentially empty) string $w=(c_0, \dots, c_{t-1}) \in \cols^*$ where $c_i \neq c_{i+1}$ for all $i < t-1$ and $c_0 \neq \mathcolor{blue}{b}$. We denote by $\mathbf W$ the set of all color words.
\end{definition}

\begin{definition}
    \label{def:closure}
    For a set $S \subset \mathbf W$ of color words, we define its \emph{closure} $\operatorname{cl} S$ as the set of prefixes of $S$, \[
        \operatorname{cl} S \coloneqq \{w \, : \, \exists z \in \cols^* \text{ s.t. } w \Vert z \in S\}
    \]
\end{definition}

\begin{definition}
    A \emph{label path} is a tuple $(w, s)$ for a color word $w \in \mathbf W$ and label $s \in \bit^{2n}$. A \emph{label set} is a set of label paths with distinct color words.
    \label{def:label-path-set}
\end{definition}

We say that a partial welded tree $T=T_{P,W}$ is left-consistent with label path $(w, s)$ if, beginning at vertex $\underline 0_0$ and following the $w_0$-th colored edge, then $w_1$-th, and so on until exhausting $w$, results in a vertex labeled by $s$. We define right consistency in the same way, except beginning at vertex $\overline 0_0$. We say that such a tree is left-consistent with a label set if it is left-consistent with every element of the label set, and similarly for right-consistent.

\begin{definition}
    A \emph{block} is defined by a contiguous subdivision of size $\sqrt{N}$ of the leaves; or more formally a set of leaves of the form $\{(m\sqrt{N})_n \dots ((m+1)\sqrt(N) - 1)_n\} \subset B_m$ for an integer $m$, or similarly for right leaves, which forms a partition $\underline B_0, \dots, \underline B_{\sqrt{N}-1}$ of left and $\overline B_0, \dots, \overline B_{\sqrt{N}-1}$ of right leaves. We assume that $\sqrt{N} = 2^{n/2}$ is a power of $2$, so each block corresponds to a distinct subtree. We interchangeably refer to the corresponding subtrees as blocks, rather than just their leaves.
    \label{def:blocks}
\end{definition}

Let $A=(\bfA, \ell_\bfA)$ and $B=(\bfB, \ell_\bfB)$ be two pairs of color word set and label function, i.e.\ $\bfA \subset \mathbf W$ and $\ell_\bfA : \bfA \rightarrow \bit^{2n}$ where $\ell_\bfA$ is injective, and similarly for $B$. Further, we always require that $\im(\ell_{\mathbf A}) \cap \im(\ell_{\mathbf B}) = \emptyset$, i.e. the functions are jointly injective.

\begin{definition}
    A partial tree $T_{P, W}$ is \emph{minimal} for $(A, B)$ as defined above if: \begin{enumerate}[nosep]
        \item $T$ is left-consistent with $A$,
        \item $T$ is right-consistent with $B$,
        \item removing a single input-output pair from $P$ or $W$ breaks condition 1 or 2.
    \end{enumerate}
    \label{def:minimal-trees}
\end{definition}

Let $T_{P,W}$ be a partial tree minimal for $(A,B)$. Regard each word in $\mathbf A$ or $\mathbf B$ as a rooted path, starting at the left or right root, respectively. For such a path $w$, let $\rho_w$ be its starting root. Its \emph{block trace} $B_1^w,B_2^w,\ldots$ records the blocks visited by the path, counting each continuous stay in a block once. Here each block
includes its entire rooted subtree. Let $p_i^w$ be the prefix
of $w$ ending when its $i$th block visit begins.
We call $B_1^w$ the \emph{initial block} of $w$; it is determined by the first $n/2$ letters of $w$.

A \emph{block collision} occurs if two visits, possibly along
the same rooted word, satisfy
\[
    B_i^w=B_j^v
    \qquad\text{and}\qquad
    (\rho_w,p_i^w)\neq(\rho_v,p_j^v).
\]
Thus a visit shared through a common rooted prefix is not a
collision. However, returning to a block after leaving it,
entering another path's initial block, or entering the same
block after the paths have diverged all count as collisions.
This includes paths that diverge within a shared block and
subsequently enter the same other block, as well as collisions
between paths starting at different roots.
We say that $w$ has \emph{block overflow} if its path exits
a block through a non-weld edge.

\begin{definition}
    For a partial tree $T=T_{P, W}$ which is minimal for $(A, B)$, the tuple $T, (A, B)$ has a \emph{block violation} if it has either a block overflow or collision.
    \label{def:block-viol}
\end{definition}

We denote by $\mathbf M(A, B)$ the set of minimal trees for $(A, B)$ with no block violations.
With these notions in hand, we can define the reachable states $\ket{A, B} \in \algo H$, where $\algo H = \algo H_{P} \otimes \bigotimes_{c\in\cols}\algo H_{W_c}$ is the Hilbert space spanned by basis states labeled by such partial injections (here taking the identification of partial welded trees with partial permutations). This will be a weighted sum of consistent, minimal $P, W$ with no block violations. In particular, we choose a weight function $w(P, W)$ as \begin{align}
    w(P, W) &\propto \Pr[P, W],
\end{align}
where the r.h.s. denotes the probability that a random injection agrees with $P, W$ on their domain and the normalization of the $w(P, W)$ is chosen such that they sum to unity, $\sum_{P,W \in \mathbf M(A, B)} w(P, W) = 1$. Then we write

\begin{align}
    \ket{A,B} &\coloneqq \sum_{P,W \in \mathbf M(A, B)} \sqrt{w(P, W)} \cdot \ket{P, W}.
\end{align}

The span of these states will be referred to as the ``color word'' subspace, as below.

\begin{definition}
    \label{def:colorword}
    The \emph{color word} subspace $\algo H_{W}$ is defined as \begin{align}
        \algo H_W &\coloneqq \spn{\ket{A, B}}.
    \end{align}
\end{definition}

We define the size of such an $A,B$ pair in terms of the amount of information gained about the weld, or more precisely about the blocks. In particular, split $\mathbf A$ into subsets $\mathbf A_{p_1}, \dots, \mathbf A_{p_q}$ indexed by prefixes $p_i \in \mathbf W$ of size $|p_i|=n/2$, or equivalently prefixes indexing blocks. The set $\mathbf A_{p_i}$ contains all color words in $\mathbf A$ whose prefix is $p_i$ (including potentially $p_i$ itself), meaning all color words entering the block corresponding to $p_i$. Similarly, split $\mathbf B$ into subsets $\mathbf B_{r_1}, \dots, \mathbf B_{r_k}$ indexed by such prefixes. With this in place, we can define the \emph{core} of $A, B$ as follows.

\begin{definition}
The \emph{core} $C(A, B)$ of $(A,B)$ is a forest whose components are
defined as follows. For each nonempty subset $\mathbf A_{p_i}$ as defined above,
let $g_{p_i}$ be the longest common prefix of its words, and let
$a_{p_i}$ be the prefix of $g_{p_i}$ of length $\min\{|g_{p_i}|,n\}$.
Consider the rooted prefix tree
$\operatorname{cl}(\mathbf A_{p_i})$, with each vertex
$w\in\mathbf A_{p_i}$ marked by its label $\ell_{\mathbf A}(w)$.
Let $T_{p_i}$ be its unique minimal connected subtree containing
$a_{p_i}$ and all marked vertices. Define the components
associated with $\mathbf B$ in the same way.

The core is the disjoint union of these components, retaining
their edge colors and marked labels, but forgetting their
original rooting and the indexing prefixes.
\label{def:core}
\end{definition}

All words in $\mathbf A_{p_i}$ have the same initial block, the block
indexed by $p_i$, which we also call the initial block of the
component $T_{p_i}$. They enter it from above and share the ordinary
path from its root to $a_{p_i}$; we call this path the component's
\emph{ordinary stem}. It has at most $n/2$ edges, and its vertices
other than $a_{p_i}$ carry no marks, weld endpoints, or branches and
are not part of the core. Components associated with $\mathbf B$ are
treated in the same way.

We can now define the size of the $(A, B)$ pair as the size of its core,
\begin{align}
    \abs{(A, B)} &\coloneqq \abs{C(A, B)}.
    \label{def:size-colorword}
\end{align}
This then naturally defines leveled subspaces that are reachable after $t$ queries, as follows.

\begin{definition}
    The $t$-query subspaces $\algo H_t$ and $\algo H_{\leq t}$ are defined by \begin{align}
        \algo H_t &= \spn{\ket{A, B} \, : \, |(A, B)| = t}, & \algo H_{\leq t} &= \spn{\ket{A, B} \, : \, |(A, B)| \leq t},
    \end{align}
    and the corresponding projectors as $\Pi_t$ and $\Pi_{\leq t}$.
    \label{def:level-subspaces}
\end{definition}

Note that $\Pi_t$ and $\Pi_{t'}$ are perpendicular for $t \neq t'$.

\subsection{Clipping and independent placements}
\label{subsec:clipping-placements}

Fix a query $(x,c)$. To extend an $(A,B)$ pair, follow the
$c$-colored edge at its $x$-labeled word: cancel the last letter if it is $c$,
and otherwise append $c$. An existing rooted prefix retains its label,
or receives a fresh label if unlabeled. A new prefix receives a fresh
label. If $x$ is absent, or the requested global-root edge does not
exist, make no change. Returning to an existing prefix is not a new
block visit; reaching its embedded vertex by a different prefix is.

Let $\mathbf M(A,B)|_{x,c}$ contain the databases in $\mathbf M(A,B)$
admitting such an extension without a block violation. For a missing
weld entry, existence of one valid partner suffices. Define
\begin{equation}
 \ket{(A,B)|^{x,c}}
 \coloneqq\sum_{(P,W)\in\mathbf M(A,B)|_{x,c}}
       \sqrt{w(P,W)}\ket{P,W},
 \label{eq:clipped-colorword}
\end{equation}
using the original weights, without renormalizing. Write
$\Pi_{\le t}|^{x,c}$ for the projector onto the span of these vectors
with $|(A,B)|\le t$.

\begin{lemma}[Clipping]
\label{lem:clipped-subspace}
For every query $(x,c)$ and integer $t\ge0$,
\begin{equation}
 \norm{\Pi_{\le t}|^{x,c}-\Pi_{\le t}}
 \le \delta_t,\qquad
 \delta_t:=64t^2N^{-1/4}.
 \label{eq:clipping-bound}
\end{equation}
\end{lemma}

We introduce the representation used both in its proof and in
\Cref{lem:subspace-preservation}. Throughout, $N=2^n$, a block has
$\sqrt N$ leaves, $M=N^2$ is the label alphabet size, and $K_c$ is
the matching size from \Cref{subsec:compressed-welded}. In particular,
$N/2\le K_c\le N$. The unnormalized extension weight is
\begin{equation}
 \omega(P,W)\coloneqq\frac1{(M)_{|P|}}
                \prod_{c\in\cols}\frac1{(K_c)_{|W_c|}},
 \qquad (a)_j\coloneqq a(a-1)\cdots(a-j+1).
 \label{eq:placement-weight}
\end{equation}
The weights $w$ defining each color-word state are proportional to
$\omega$ on its support.

\paragraph{Fixing the core.}
The database determines its core: in each block, join its labeled
vertices and stored weld endpoints by their minimal ordinary subtree,
then join these subtrees by the stored welds. On $\mathbf M(A,B)$ this
is precisely $C(A,B)$, with embedded positions forgotten. Each component
contains a mark and has at most one unmarked leaf. If present, this leaf lies in the initial block and its only core edge
is a weld: the words descend to it from above and then cross that
weld. We call it the component's \emph{entry leaf}. These facts are proved in
\Cref{app:core-representation}.

A sector $\sigma$ fixes the labels outside the blocks and the abstract
marked, edge-colored core forest $C_1,\ldots,C_q$. It does not fix which
core edges are welds. Let $\mathcal H_\sigma$ be its color-word span.
Different sectors have disjoint database supports, and
\begin{equation}
 \mathcal H_{\le t}
 =\bigoplus_{\sigma:\,\sum_i|C_i|\le t}\mathcal H_\sigma.
 \label{eq:placement-sectors}
\end{equation}
The same decomposition holds for the clipped spaces. Labels above the
blocks are fixed within a sector and do not contribute to its size.

\paragraph{Independent placements.}
Let $\Omega$ be the embedded vertices in blocks, including block roots.
For $c\in\cols$ put
\begin{equation}
 Q_c(u,v) \coloneqq
 \begin{cases}
  1&uv\text{ is an ordinary $c$-edge inside a block},\\
  K_c^{-1}&u,v\text{ are eligible leaves on opposite sides},\\
  0&\text{otherwise}.
 \end{cases}
 \label{eq:placement-kernel}
\end{equation}
For a placement $\phi_i:C_i\to\Omega$, set
\begin{equation}
 \mu_i(\phi_i)=\frac{\theta_i(\phi_i)}{Z_i}
       \prod_{ab\in E(C_i)}Q_{c(ab)}(\phi_i(a),\phi_i(b)),
 \qquad X_{i,a}=\phi_i(a),
 \label{eq:placement-law}
\end{equation}
where $Z_i$ normalizes the measure. If all leaves are marked, set
$\theta_i=1$ and allow any core vertex as an anchor. Otherwise let
$a_i$ be its entry leaf: require its image to be eligible
for its incident edge, and allow only $a_i$ as an anchor. Fixing an
anchor means fixing its assigned position. This law permits block
collisions but excludes ordinary exits through block roots.

For a random variable $X$, let $L^2(X)$ be the functions of $X$, with
the inner product inherited from the indicated probability measure.
Writing $I_i=C_i$ in the first case and $I_i=\{a_i\}$ in the second,
define
\begin{equation}
 \mathcal K_i\coloneqq\sum_{a\in I_i}L^2(X_{i,a}),\qquad
 \mu\coloneqq\bigotimes_i\mu_i,\qquad
 \mathcal K_\sigma\coloneqq\bigotimes_i\mathcal K_i.
 \label{eq:placement-space}
\end{equation}
These are spaces of arbitrary superpositions of anchor functions.

\paragraph{Returning to databases.}
A placement is good when distinct ordinary components, within or
between cores, occupy distinct blocks. Let $G_\sigma$ multiply by this
indicator. A good placement gives a unique database $D_\phi=(P_\phi,W_\phi)$.
If $e_c(\phi)$ counts its weld edges of color $c$, put
\begin{align}
 S_{0,\sigma}(\phi)
 &\coloneqq\prod_c\left(\frac{K_c^{e_c(\phi)}}{(K_c)_{e_c(\phi)}}\right)^{1/2},
 \nonumber\\
 S_\sigma&\coloneqq\frac{S_{0,\sigma}}{\max_{\phi\text{ good}}S_{0,\sigma}}
 \quad\text{on good placements},\qquad M_\sigma=G_\sigma S_\sigma.
 \label{eq:placement-map}
\end{align}
Set $S_\sigma=1$ off the good set. Identify a function supported on
good placements isometrically with
$\sum_\phi\sqrt{\mu(\phi)}f(\phi)\ket{D_\phi}$.

\begin{lemma}[Placement representation]
\label{lem:placement-lift}
The sector span is exactly $\mathcal H_\sigma=M_\sigma\mathcal K_\sigma$.
For each core, every $f_i\in\mathcal K_i$ has a decomposition
\begin{equation}
 f_i=\sum_{a\in I_i}f_{i,a},\qquad
 f_{i,a}\in L^2(X_{i,a}),\qquad
 \sum_a\norm{f_{i,a}}^2\le\norm{f_i}^2.
 \label{eq:anchor-decomposition}
\end{equation}
Suppose $\sum_i|C_i|\le t$ and $t^2N^{-1/4}\le1/64$.
There is a subspace $\mathcal K_{\sigma,0}\subseteq\mathcal K_\sigma$
such that every $h\in\mathcal H_\sigma$ has an 
$f$ satisfying
\begin{equation}
 M_\sigma f=h,\qquad \norm f\le2\norm h,\qquad
 \norm{(I-M_\sigma)f}\le\eta_t\norm f,
 \quad
 \eta_t=4t^2N^{-1/4}+t^2/N,
 \label{eq:placement-lift}
\end{equation}
and $\norm{I-S_\sigma}\le t^2/N$ on good placements.
All statements allow arbitrary fixed labels outside the blocks and tensoring with
untouched registers.
\end{lemma}

The discarded subspace consists of combinations of anchors that force
two cores into the same block. 
\Cref{sec:appendix} proves both lemmas. We use the same placement model
there for two events: hitting a specified occupied block and reaching
a block root after crossing a weld.

\section{Main lower bound}
\label{sec:main_lower}

In this section, we prove a query lower bound for path-finding in the welded tree. In particular, we show that queries to the compressed welded oracle approximately maintain the color-word subspace, and further that this subspace almost certainly does not contain a left-to-right path.

To begin, we modify weld compression so that its non-bot vector excludes opposite
endpoints that are already labeled in $P$. We denote the resulting sanitized compression by $\widehat{\overline{\wC}}$, which by \Cref{lem:restricted-compression} has difference $O(t^{1/2} N^{-1/2})$ with the standard sanitized $\overline{\wC}$ on databases of core size at most $t$. It will also be useful to define a weld compression operator which avoids any non-empty block, where a block is empty if it is incident on a weld edge or has any marked vertex. We denote this weld compression operator by $\overline{\wC}^{\,0}$.

\begin{lemma}[Subspace preservation]
\label{lem:subspace-preservation}
Let $d=n/2+1$ and $T=t+d$.
Then, uniformly in the query registers,
\begin{align}
 \norm{\Pi_{\le T}^{\perp}
       {\widehat{\overline{\wC}}}^\dagger\Pi_{\le t}}
 &=O(T^2N^{-1/4}),\\
 \norm{\Pi_{\le T}^{\perp}
       {\widehat{\overline{\wC}}}\Pi_{\le t}}
 &=O(T^2N^{-1/4}).
\end{align}
\end{lemma}

\begin{proof}
Fix $(x,c)$. The query acts on the $c$-neighbor of the vertex labeled $x$.
That neighbor is either outside the core (i.e., fresh), the entry leaf of 
its component, a labeled core vertex, or an unlabeled internal core vertex.
The first two cases are handled by the placement argument below; 
the others are handled directly in the paragraphs below labeled 
``Compression'' and ``Decompression.'' Queries that cross
a block boundary or stay outside the blocks are treated at the end.

We may assume $T^2N^{-1/4}$ is sufficiently small:
otherwise contractivity proves the bounds after enlarging the constant.
Write $\varepsilon=O(T^2N^{-1/4})$ for the uniform errors below.

\paragraph{Local formulas.}
Let $V$ denote the compressed oracle operator which creates a normalized uniform missing entry: in particular for a given query point $x$ with valid image set $S$, the operator is $V = \ket{+_S}\bra{\bot}$. Then $P_0=V^\dagger V$
projects onto the missing-entry space, the sanitized local operators are
\begin{equation}
 \overline C^\dagger=I-P_0-VV^\dagger+V,
 \qquad
 \overline C=I-P_0-VV^\dagger+V^\dagger.
 \label{eq:pres-local-formulas}
\end{equation}
Let $L$ denote uniform creation of the neighbor's label, in the sense of $V$ above. At a high level, its coefficients depend only
on the number of remaining labels, not on their embedded positions.
Thus averaging a retained label will preserve the color-word span, up to minor corrections due to clipping.

It is convenient to require that a newly selected weld endpoint's entire
block be empty after deleting the queried pair. For a residual matching
$J$, let $Y_0$ be its eligible endpoints with this property, and set
\begin{equation}
 E\ket J=\frac1{\sqrt{K_c-|J|}}
       \sum_{v\in Y_0}\ket{J[u\mapsto v]},\qquad
 E_0\ket J=\frac1{\sqrt{|Y_0|}}
       \sum_{v\in Y_0}\ket{J[u\mapsto v]}.
 \label{eq:pres-creation}
\end{equation}
Here the other database records, including $P(u)=x$, are suppressed;
right-to-left queries use the inverse matching. Use $E_0$ in
\eqref{eq:pres-local-formulas}, and denote the resulting combined
compression by $\overline{\wC}^{\,0}$. At most $T$ blocks are occupied,
so at most $T\sqrt N$ endpoints are excluded. On such database
supports, comparison of the uniform vectors gives
\begin{align}
 \norm{(\widehat{\overline{\wC}}-\overline{\wC}^{\,0})\psi}
 &\le O(\sqrt T\,N^{-1/4})\norm\psi,\nonumber\\
 \norm{E-E_0}&=O(T/\sqrt N).
 \label{eq:pres-normalizations}
\end{align}
The first bound also holds with both operators adjointed.
Top labels exclude no leaf endpoints. We prove
preservation for the empty-block operators and transfer back at the end.

\paragraph{Fresh extension and erasure at an entry leaf.}
Call the queried neighbor fresh when it is outside the current core.
Let $\mathcal G$ be the span of the corresponding clipped descriptions,
with a fixed remaining marked core forest. Let $A$ expose the neighbor
without labeling it: it is $E$ for a weld and the identity for an
ordinary edge. Uniform extension, with the weights
\eqref{eq:placement-weight}, gives
\begin{equation}
 LA\mathcal G\subseteq\mathcal H_W,\qquad
 \norm{Ag}^2\ge(1-O(T/\sqrt N))\norm g^2.
 \label{eq:pres-fresh}
\end{equation}
The first inclusion is exact on the whole span: extending a clipped
description and summing over its fresh label gives its extended
color-word states. The matching and label factors are respectively
$(K_c-|J|)^{-1/2}$ and $(M-|P|)^{-1/2}$.

The other possibility is that the queried neighbor is the entry leaf $a$
of its component. It is unlabeled, and its
only core edge is the stored $c$-weld to the marked vertex carrying
$x$. Write $\mathcal R$ for the span of the color-word states of this
form. We show that there is a space
$\mathcal F\subseteq\mathcal G\cap\mathcal H_W$ with
\begin{equation}
 A\mathcal F\subseteq\mathcal R,\qquad
 \operatorname{dist}(A^\dagger r,\mathcal F)
       \le\varepsilon\norm r\quad(r\in\mathcal R).
 \label{eq:pres-root-erasure}
\end{equation}

Use the independent placements of
\Cref{lem:placement-lift} for the source and for the forest obtained
by deleting $a$. In the source, the distinguished core supplies only
$L^2(X_a)$; after deletion all its leaves are marked. In orthonormal
formal-placement coordinates, let $R_0$ erase the endpoint with
coefficient $K_c^{-1/2}$. Each row has at most $K_c$ entries and different
rows have disjoint supports, so $\norm{R_0}\le1$. The weld kernel is
constant on each direction across the weld. Consequently, summing out
$X_a$ sends a function of $X_a$ to a constant on the appropriate
eligible values of $X_x$, times the untouched factors. In particular,
\begin{equation}
 R_0\mathcal K_{\rm src}\subseteq\mathcal F_0
                 \subseteq\mathcal K_{\rm dst},
 \label{eq:pres-ideal-erasure}
\end{equation}
where $\mathcal F_0$ places the remaining distinguished component
freely, with $x$ at an eligible leaf. This is an inclusion of spaces,
so it allows entanglement with every other component.

Set $\mathcal F=M_{\rm dst}\mathcal F_0$. For $r\in\mathcal R$ choose
$f\in\mathcal K_{{\rm src},0}$ with
$M_{\rm src}f=r$ and $\norm f\le2\norm r$. Erasure cannot introduce
a collision. With the constant erasure coefficient, the difference
from the particular vector $M_{\rm dst}R_0f\in\mathcal F$ is
\begin{align}
 G_{\rm dst}R_0M_{\rm src}f-M_{\rm dst}R_0f
 ={}&G_{\rm dst}R_0(M_{\rm src}-I)f\nonumber\\
    &+G_{\rm dst}(I-S_{\rm dst})R_0f.
 \label{eq:pres-erasure-transfer}
\end{align}
Its norm is at most $(\eta_T+T^2/N)\norm f$ by
\eqref{eq:placement-lift}. Replacing $K_c^{-1/2}$ by
$(K_c-|J|)^{-1/2}$ costs $O(T/K_c)\norm r$: on each residual-database
row it changes the coefficient by that relative factor.
This proves the distance estimate in \eqref{eq:pres-root-erasure}.
The missing $c$-port of $x$ is a weld port, so every description in
$\mathcal F$ admits a fresh extension and is already clipped.
Reattaching a uniformly chosen entry leaf across the $c$-weld gives a sum of descriptions
anchored to it. The exact matching-weight ratio is
$(K_c-|J|)^{-1/2}$, proving $A\mathcal F\subseteq\mathcal R$.

\paragraph{Controlling the two descriptions of an exposed neighbor.}
The spaces $\mathcal R$ and $A\mathcal G$ can overlap. Put
\begin{equation}
 \mathcal I=\mathcal R\cap A\mathcal G,
 \qquad \mathcal R_0=\mathcal R\cap\mathcal I^\perp.
\end{equation}
For $r\in\mathcal R_0$, the inclusion $A\mathcal F\subseteq\mathcal I$
gives $A^\dagger r\perp\mathcal F$. Hence
\eqref{eq:pres-root-erasure} implies
$\norm{A^\dagger r}\le\varepsilon\norm r$.
Together with \eqref{eq:pres-fresh}, this bounds the cross term in
$\norm{r+Ag}^2$. Every exposed vector therefore has a decomposition
\begin{equation}
 \xi=r+Ag,\qquad r\in\mathcal R_0,\quad g\in\mathcal G,
 \qquad \norm r+\norm g=O(\norm\xi).
 \label{eq:pres-stable}
\end{equation}
Indeed, first absorb the $\mathcal I$ component into $A\mathcal G$;
then use
$|\langle r,Ag\rangle|\le\varepsilon\norm r\norm g$.
This is the only passage from the two descriptions to their combined
span; no bound on arbitrary color-word coefficients is needed.

\paragraph{Compression.}
First consider an exposed neighbor whose label is defined. Sanitized
label compression is $I-LL^\dagger+L^\dagger$. The first two terms
retain a label, so empty-block weld compression fixes their edge;
relabeling preserves the color-word span. Analyze
$\xi=L^\dagger\psi$, noting $\norm\xi\le\norm\psi$.

If the neighbor remains an internal core vertex, its other incident
information protects its weld edge, and removing the mark gives a
color-word state. Otherwise it is terminal. The anchor description
from \Cref{lem:placement-lift} gives exactly
$\xi\in A\mathcal G+\mathcal R$: deleting a terminal marked word
gives a fresh clipped description, except when removing the mark makes the neighbor the entry leaf of its component. That case gives $\mathcal R$. An ordinary terminal
endpoint belongs to the fresh case, including when, after the deletion, it lies on the
ordinary stem of its component (this happens when it was the vertex $a_{p_i}$). The initial-anchor position determines whether
its incident edge is ordinary or a weld, so this argument does not
project a state onto an arbitrary weld pattern.

On a fresh vector, weld compression sends $Ag$ to $g$ with error
$O(T/\sqrt N)\norm g$. Since $g$ is clipped, its distance to the
ordinary space of the same level is at most $\delta_T\norm g$.
On $r\in\mathcal R$, its action is
\begin{equation}
 r-E_0E_0^\dagger r+E_0^\dagger r
       =r-Af+f+e,
 \qquad f\in\mathcal F,\quad \norm e\le\varepsilon\norm r,
 \label{eq:pres-compress-root}
\end{equation}
by \eqref{eq:pres-root-erasure} and \eqref{eq:pres-normalizations}. All three main
terms are color-word vectors. Apply \eqref{eq:pres-stable} to bound
the error relative to $\norm\xi$. On the other active branches an
undefined neighbor label is annihilated by sanitized compression.

\paragraph{Decompression.}
For a fresh clipped input, the output is $LAg$ up to the normalization
error in \eqref{eq:pres-normalizations}; use \eqref{eq:pres-fresh}.
An already labeled neighbor gives $I-LL^\dagger$, which only averages
its label. An unlabeled internal neighbor is protected, and adding
its mark gives a color-word state. On $r\in\mathcal R$, sanitization
removes the erased-entry term, leaving
\begin{equation}
 L(r-E_0E_0^\dagger r)
       =Lr-LAf+e,\qquad f\in\mathcal F,\quad
       \norm e\le\varepsilon\norm r.
\end{equation}
Both main terms mark the entry leaf of a vector in
$\mathcal R$, so they are color-word vectors.

\paragraph{Boundary queries, levels, and uniformity.}
Apply the preceding calculation to the clipped input; replacing an
ordinary input by its clipped projection costs at most $\delta_t$.
For a query entering a block from above, the clipped descriptions
either already use that initial block or have it empty. Adding or
removing the neighbor mark therefore gives the corresponding
color-word description. For a query leaving a block through its top,
clipping retains that event only through the safe initial-block span
identified in \Cref{app:clipping-events}; its action is again ordinary
label creation or removal. The event projectors preserve these safe
spans, so the boundary calculation holds on their full spans.
Queries wholly outside the blocks change only top labels.

For the label-erasure calculation, fix the remaining labels and the
core obtained by temporarily marking the exposed neighbor. These
data are determined by the database and $(x,c)$; different choices
have disjoint supports. Uniform label averaging is a contraction
between these spaces. Within each such space the only overlapping
alternatives are $A\mathcal G$ and $\mathcal R$, controlled by
\eqref{eq:pres-stable}. Thus the estimates apply to arbitrary
superpositions, including coherent query registers.

An extension inside a block adds at most one core vertex. Entering an
occupied initial block from above can additionally add the component's
ordinary stem to the core, adding at most
$n/2+1=d$ vertices. Erasure cannot enlarge the core, and resampling
replaces one terminal vertex by one. Every approximation above thus
lies in $\mathcal H_{\le T}$. 
Combining the clipping errors $\delta_t$ and $\delta_T$ from
\eqref{eq:clipping-bound}, the preimage error in
\eqref{eq:pres-erasure-transfer} (which gives the $\varepsilon$ of
\eqref{eq:pres-root-erasure}), and the normalization errors
\eqref{eq:pres-normalizations} gives $O(T^2N^{-1/4})$ in both
directions, and the same comparison transfers back to
$\widehat{\overline{\wC}}$.
\end{proof}

\begin{lemma}
\label{lem:reachable-colorword}
Let $\algo A$ denote a quantum algorithm querying the welded tree oracle.
Initialize a welded tree
database with the root label as above. Let $\ket{\psi_q}$
be the joint state after $q\geq 1$ queries to the ordinary
compressed welded tree oracle $\cwO$.
Put
\[
    d\coloneqq n/2+1,\qquad S\coloneqq 2dq=(n+2)q,\qquad
    K_{\min}\coloneqq \min_{c\in\cols}K_c,
\]
and define
\[
    \delta_S
    \coloneqq
    \max_{\substack{s\leq S\\x,c}}
    \norm{\Pi_{\leq s}|^{x,c}-\Pi_{\leq s}}.
\]
Let $Q=O(q+1)$ bound the number of component permutation calls,
including initialization, and suppose
$Q\leq \frac12\min\{M,K_{\min}\}$. Write
\[
    \Delta_Q
    =
    \sigma_{\mathsf{perm}}(Q,M)
    +\sum_{c\in\cols}\sigma_{\mathsf{perm}}(Q,K_c)
    +q\sqrt{\frac{K}{M-Q}}
    +q\sqrt{\frac{Q}{K_{\min}-Q}}.
\]
Then
\[
    \norm{
        (I_A\otimes\Pi_{\leq S}^{\perp})\ket{\psi_q}
    }
    =
    O\left(
        \Delta_Q
        +\frac{qS^2}{N^{1/4}}
    \right).
\]
\end{lemma}

\begin{proof}
Purify the algorithm, so its operations between queries are
unitaries on its workspace and query registers. We suppress
identity operators on these registers when writing database
projectors. All sanitized and projected states below remain
subnormalized.

We first replace the ordinary compressed experiment by the
sanitized experiment. This comparison must be made in joint
state norm, before tracing out the databases.

Consider the componentwise compressed evaluation circuit used
in the proof of \Cref{cor:sanitized-weld-simulation}.
Replacing its inverse label calls by direct database lookup
gives the $\cwO$ experiment and costs
\[
    O\left(q\sqrt{\frac{K}{M-Q}}\right)
\]
in joint state norm, by \Cref{lem:inverse-sanitization}.
Alternatively, sanitize its component calls and make the same
inverse-label replacements. The resulting circuit is the
$\overline{\cwO}$ experiment, and its distance from the
componentwise compressed circuit is at most
\[
    \sigma_{\mathsf{perm}}(Q,M)
    +\sum_{c\in\cols}\sigma_{\mathsf{perm}}(Q,K_c)
    +O\left(q\sqrt{\frac{K}{M-Q}}\right).
\]
Indeed, use hybrids with ordinary component calls before each
replacement and modified calls afterwards. Reachable validity
applies to each ordinary prefix, while every modified suffix
is a contraction. Thus validity is never invoked on a state
that has already undergone sanitization or projection.
The triangle inequality gives the same bound, up to constants,
between the $\cwO$ and $\overline{\cwO}$ experiments.

Next impose the label protection from
\Cref{lem:subspace-preservation}. Set
\[
    C=\widehat{\overline{\wC}},
    \qquad
    \widehat{\overline{\cwO}}
    =C\wP C^\dagger.
\]
On the databases encountered in an ordinary sanitized prefix,
each relevant database has at most $Q$ entries.
Excluding already labeled opposite endpoints removes at most
$Q$ terms from a weld uniform vector containing at least
$K_c-Q$ terms. Comparing the normalized uniform vectors,
and hence their compression operators as in \Cref{lem:restricted-compression}, costs
\[
    O\left(\sqrt{\frac{Q}{K_c-Q}}\right).
\]
The same estimate holds after sanitization.
A hybrid over the two weld compression factors in each query
therefore costs
\[
    O\left(q\sqrt{\frac{Q}{K_{\min}-Q}}\right).
\]
Here the prefixes use the ordinary sanitized oracle, whose
database sizes are bounded by $Q$, and the protected suffixes
are contractions.

Let $\ket{\widehat\psi_q}$ be the resulting protected,
sanitized state, with sanitized initialization. We have proved
\begin{equation}
\label{eq:reachable-sanitization-error}
    \norm{\ket{\psi_q}-\ket{\widehat\psi_q}}
    =O(\Delta_Q).
\end{equation}

We now bound the leakage of one protected sanitized query.
We first record the effect of database lookup.
For fixed $(x,c)$, lookup acts uniformly on each clipped
color-word generator: it returns the prescribed label of the
neighboring marked word, if present, and otherwise gives no
answer. Clipping excludes an identification with a different
marked word. Thus lookup preserves the clipped subspace,
without increasing its level.

More formally, let
\[
    P_s^{\mathrm{cl}}
    \coloneqq
    \sum_{x,c}\proj{x,c}\otimes\Pi_{\leq s}|^{x,c},
    \qquad
    P_s \coloneqq I_A\otimes\Pi_{\leq s},
\]
with identities on the other registers understood.
Since $\wP$ is unitary and preserves the clipped subspace,
\[
    [\wP,P_s^{\mathrm{cl}}]=0.
\]
Consequently,
\begin{equation}
\label{eq:reachable-lookup-leakage}
    \norm{P_s^\perp\wP P_s}
    \leq 2\norm{P_s-P_s^{\mathrm{cl}}}
    \leq 2\delta_S
    \qquad(s\leq S).
\end{equation}
This also accounts for the fact that a clipped vector need
not itself belong to the ordinary color-word subspace.

By \Cref{lem:subspace-preservation}, uniformly for the levels
used below,
\[
    \norm{P_{r+d}^{\perp}C^\dagger P_r}=O\Bigl(\frac{S^2}{N^{1/4}}\Bigr),
    \qquad
    \norm{P_{r+d}^{\perp}C P_r}=O\Bigl(\frac{S^2}{N^{1/4}}\Bigr).
\]
For $r+2d\leq S$, insert $P_{r+d}$ after decompression and
after lookup. Since $C$, $C^\dagger$, and $\wP$ are
contractions, \eqref{eq:reachable-lookup-leakage} gives
\begin{align*}
    \norm{P_{r+2d}^{\perp}
        \widehat{\overline{\cwO}}P_r}
    &=
    \norm{P_{r+2d}^{\perp}C\wP C^\dagger P_r}\\
    &\leq
    \norm{P_{r+d}^{\perp}C^\dagger P_r}
    +\norm{P_{r+d}^{\perp}\wP P_{r+d}}
    +\norm{P_{r+2d}^{\perp}C P_{r+d}}\\
    &=O\Bigl(\frac{S^2}{N^{1/4}}\Bigr).
\end{align*}

Finally, hybrid over the $q$ protected sanitized queries,
projecting onto level at most $2dj$ after query $j$.
Sanitized initialization is supported on databases containing
only root labels, so its state $\ket{\phi_0}$ lies in
$\mathcal H_{\leq 0}$ and has norm at most one.
Writing $U_j$ for the algorithm's intervening unitaries,
define the projected evolution by
\[
    \ket{\phi_j}
    =
    P_{2dj}\,
    \widehat{\overline{\cwO}}\,
    U_j\ket{\phi_{j-1}}.
\]
Workspace unitaries commute with the database projectors.
Moreover, every factor in this evolution is a contraction,
so $\norm{\phi_j}\leq 1$ and
$\ket{\phi_j}\in\mathcal H_{\leq 2dj}$.

Comparing this evolution with the unprojected protected
sanitized evolution, the preceding one-query estimate yields
\[
    \norm{\ket{\widehat\psi_j}-\ket{\phi_j}}
    \leq
    \norm{\ket{\widehat\psi_{j-1}}-\ket{\phi_{j-1}}}
    +O(\frac{S^2}{N^{1/4}}).
\]
Induction therefore gives
\[
    \norm{\ket{\widehat\psi_q}-\ket{\phi_q}}
    =O(q\frac{S^2}{N^{1/4}}).
\]
Since $P_S^\perp\ket{\phi_q}=0$, we conclude from
\eqref{eq:reachable-sanitization-error} that
\begin{align*}
    \norm{P_S^\perp\ket{\psi_q}}
    &\leq
    \norm{\ket{\psi_q}-\ket{\widehat\psi_q}}
    +\norm{\ket{\widehat\psi_q}-\ket{\phi_q}}\\
    &=O\Bigl(\Delta_Q+q\frac{S^2}{N^{1/4}}\Bigr),
\end{align*}
as claimed.
\end{proof}

\begin{theorem}[Path-finding lower bound]
\label{thm:path-finding}
Let $N=2^n$, with $n$ even, and sample the independently matched
welded tree ensemble of \Cref{subsec:compressed-welded}.
An algorithm receives $P(\underline 0_0)$ and makes $q\ge1$
quantum queries to the welded tree oracle. Let $p_{\rm path}$
be the probability that it outputs an ordered sequence of labels forming a path from $P(\underline 0_0)$ to $P(\overline 0_0)$.
There is no restriction on the length of its output.
Let $S=(n+2)q$, let $Q=O(q+1)$ be a sufficiently large bound
on component queries, and suppose $Q\le\min\{M,K_{\min}\}/2$.
With $\Delta_Q$ as in \Cref{lem:reachable-colorword}, define
\begin{align*}
 E_Q&\coloneqq \varepsilon_{\mathsf{perm}}(Q,M)
       +\sum_{c\in\cols}\varepsilon_{\mathsf{perm}}(Q,K_c),\\
 V_Q&\coloneqq Q\sqrt{\frac{K}{M-Q}}
       +Q\sqrt{\frac{Q}{K_{\min}-Q}}.
\end{align*}
Then
\begin{equation}
 p_{\rm path}
 \le O(E_Q+V_Q)
   +O\left(\Delta_Q+\frac{qS^2}{N^{1/4}}\right)^2.
 \label{eq:path-general-bound}
\end{equation}
In particular, substituting
$\varepsilon_{\mathsf{perm}}(Q,L)=O(QL^{-1/2})$ and
$\sigma_{\mathsf{perm}}(Q,L)=O(Q^2L^{-1/2})$ we obtain
\begin{equation}
 p_{\rm path}
 \le O\left(\frac{(n+2)^4q^6}{N^{1/2}}\right).
 \label{eq:path-query-bound}
\end{equation}
Consequently, constant success probability requires
$q=\Omega\left(N^{1/12}/(n+2)^{2/3}\right)=2^{\Omega(n)}$.
\end{theorem}

\begin{proof}
Call a path \emph{recorded} if all its vertices have labels in
$P$ and all its weld edges occur in the corresponding $W_c$.
The ordinary tree edges are fixed and require no database entry.
Let $\Gamma$ be the diagonal projector onto databases containing
a recorded path between the two roots.

\paragraph{Recorded paths in the color-word space.}
We first show
\begin{equation}
 \norm{\Gamma\Pi_{\le s}}=O(s^2N^{-1/4}).
 \label{eq:path-colorword-overlap}
\end{equation}
A recorded path has a segment which enters the blocks from the
left upper tree and next returns to the upper trees on the right.
This segment lies inside the blocks, including its two endpoints,
which are distinct block roots. In a good placement, different
ordinary components occupy distinct blocks. Every such segment
therefore lies in one core: its ordinary edges lie in the terminal
hulls, and its weld edges join those hulls.
Thus $\Gamma$ implies that some core contains two distinct block roots.
We bound this event using the placement representation. It suffices to
consider $s\ge1$ with $s^2N^{-1/4}\le1/64$; zero-core sectors have no
recorded path, and outside this range the claimed bound is trivial.
In a sector with core sizes $t_i$ summing to at most $s$, let $B_i$
be the event that core $i$ contains two distinct block roots. Given
any allowed anchor, its initial ordinary component contains at most
one block root. Hence $B_i$ requires some vertex reached after a weld
to be a block root. By estimate (iv) in
\Cref{app:exploration-estimates} and a union bound over the $t_i$
vertices, $\Pr(B_i\mid\text{anchor})=O(t_i/\sqrt N)$.
Applying \eqref{eq:anchor-decomposition} and
\eqref{eq:app-event-restriction} to at most $t_i$ anchor coordinates
gives $\norm{\ind_{B_i}\Pi_{\mathcal K_i}}^2=O(t_i^2/\sqrt N)$.
Tensoring with the other cores and union bounding therefore gives
\[
 \norm{\ind_{\bigcup_i B_i}f}^2
 \le O(s^2/\sqrt N)\norm f^2
 \qquad(f\in\mathcal K_\sigma).
\]
For $h\in\mathcal H_\sigma$, choose 
$M_\sigma f=h$ with $\norm f\le2\norm h$ from
\Cref{lem:placement-lift}. Since $M_\sigma$ is a contraction commuting
with these event indicators and $\Gamma$ implies $\bigcup_i B_i$,
$\norm{\Gamma h}=O(sN^{-1/4})\norm h$.
Taking the maximum over the orthogonal sectors proves
\eqref{eq:path-colorword-overlap}.

\paragraph{From an output path to a recorded path.}
We claim that, for the ordinary compressed state $\ket{\psi_q}$
from \Cref{lem:reachable-colorword},
\begin{equation}
 p_{\rm path}\le\norm{\Gamma\ket{\psi_q}}^2+O(E_Q+V_Q).
 \label{eq:path-recording}
\end{equation}
The following is a fundamental-lemma argument, with an explicit
truncation to avoid charging for a potentially long output.
Erase loops from the output using classical computation, so its
labels are distinct. Append a verifier with access to the component
oracles. If the output contains at least $q+2$ labels, it checks
only that the first $q+2$ labels belong to $P(V)$. Otherwise it
decodes every output label with $P^{-1}$, checks that the endpoints
are the two roots, and verifies all consecutive edges. Ordinary
edges are checked without queries; weld edges are checked by
forward queries to the eligible $W_c$, orienting each pair from
left to right. A genuine path always passes this verifier, which
uses only $O(q+1)$ component queries.

Replace the component oracles by compressed oracles. Soundness
costs $E_Q$ in acceptance probability. Replace all inverse label
calls, including those in the verifier, by direct database lookup.
By \Cref{lem:inverse-sanitization}, the joint-state error is
$O(Q\sqrt{K/(M-Q)})$. The algorithm's prefix is now exactly the
$\cwO$ experiment, and the verifier's label lookups do not change
$P$. Since $|P|\le q+1$, its long-output branch always rejects.

On the remaining branch, keep the decoded vertices in the
verifier's registers. In each verification call to $W_c$, modify
compression to exclude the decoded right leaves from its uniform
extension vector. There are at most $q+1$ such leaves.
\Cref{lem:restricted-compression} bounds the total change by
$O(Q\sqrt{Q/(K_{\min}-Q)})$ in joint-state norm.
These modified operators fix every already recorded pair ending
at a decoded right leaf and cannot create a new such pair.
Consequently, a checked weld edge can pass only if it was already
recorded before verification. The modified verifier accepts only
on $\Gamma$, proving \eqref{eq:path-recording}.

\paragraph{Applying Lemma 4.2.}
By \Cref{lem:reachable-colorword}, the clipping bound, and the
revised \Cref{lem:subspace-preservation},
\[
 \norm{\Pi_{\le S}^{\perp}\ket{\psi_q}}
 =O\left(\Delta_Q+qS^2N^{-1/4}\right).
\]
Together with \eqref{eq:path-colorword-overlap}, this gives
\[
 \norm{\Gamma\ket{\psi_q}}
 \le \norm{\Pi_{\le S}^{\perp}\ket{\psi_q}}
      +\norm{\Gamma\Pi_{\le S}}
 =O\left(\Delta_Q+qS^2N^{-1/4}\right).
\]
Equation~\eqref{eq:path-recording} proves
\eqref{eq:path-general-bound}. Substitute $M=N^2$, $K=O(N)$,
$K_c\ge N/2$, and the stated permutation estimates.
Then $E_Q+V_Q=O(q^{3/2}/\sqrt N)$ and
$\Delta_Q=O(q^2/\sqrt N)$, while
$qS^2N^{-1/4}=(n+2)^2q^3N^{-1/4}$.
Substituting these estimates into \eqref{eq:path-general-bound}
gives \eqref{eq:path-query-bound}. Constant success probability
therefore requires $q=\Omega(N^{1/12}/(n+2)^{2/3})$.
Since $n=\log_2N$, this is $2^{\Omega(n)}$.
\end{proof}

\section*{AI statement}

The general proof strategy described in this paper was developed by the authors without the use of AI. GPT-6 was used to develop proof ideas in the later stages of the work, and in particular, identified an error in a human-generated proof of the clipping lemma, which it then corrected. The manuscript was written and revised with significant input from GPT-6 and additional input from Claude Opus 5, and was edited by the authors, who take sole responsibility for its correctness.

\section*{Acknowledgments}

This work received support from the National Science Foundation (grant 26-17356) and the Department of Energy (grant DE-SC0020264; the Office of Science, Office of Advanced Scientific Computing Research, Accelerated Research in Quantum Computing program; and the Office of Science, National Quantum Information Science Research Centers, Quantum Systems Accelerator (QSA)).

\newpage

\printbibliography

@misc{cm26comp,
  author = {Carolan, Joseph and Majenz, Christian},
  title = {Compressed Permutation Oracles Revisited (Upcoming Work)},
  year = {2026}
}

@inproceedings{childs03weld,
author = {Childs, Andrew M. and Cleve, Richard and Deotto, Enrico and Farhi, Edward and Gutmann, Sam and Spielman, Daniel A.},
title = {Exponential algorithmic speedup by a quantum walk},
year = {2003},
publisher = {Association for Computing Machinery},
address = {New York, NY, USA},
doi = {10.1145/780542.780552},
booktitle = {Proceedings of the Thirty-Fifth Annual ACM Symposium on Theory of Computing},
pages = {59–68},
numpages = {10},
series = {STOC '03}
}

@inproceedings{carolan2025compressedpermutationoracles,
author = {Carolan, Joseph},
title = {Compressed Permutation Oracles},
year = {2026},
doi = {10.1145/3798129.3800736},
booktitle = {Proceedings of the 58th Annual ACM Symposium on Theory of Computing},
pages = {150–161},
      eprint={2509.18586},
      archivePrefix={arXiv}
}

@InProceedings{carolan24oneway,
author="Carolan, Joseph
and Poremba, Alexander",
title="Quantum One-Wayness of the Single-Round Sponge with Invertible Permutations",
booktitle="Advances in Cryptology -- CRYPTO 2024",
year="2024",
pages="218--252"
}

@article{Shor_1997,
   title={Polynomial-Time Algorithms for Prime Factorization and Discrete Logarithms on a Quantum Computer},
   volume={26},
   DOI={10.1137/s0097539795293172},
   number={5},
   journal={SIAM Journal on Computing},
   publisher={Society for Industrial & Applied Mathematics (SIAM)},
   author={Shor, Peter W.},
   year={1997},
   pages={1484–1509} }

@misc{cpz24precomp,
      author = {Joseph Carolan and Alexander Poremba and Mark Zhandry},
      title = {(Quantum) Indifferentiability and Pre-Computation},
      howpublished = {Cryptology {ePrint} Archive, Paper 2024/1727},
      year = {2024},
      url = {https://eprint.iacr.org/2024/1727}
}

@InProceedings{CHS19,
	author="Czajkowski, Jan
	and H{\"u}lsing, Andreas
	and Schaffner, Christian",
	editor="Boldyreva, Alexandra
	and Micciancio, Daniele",
	title="Quantum Indistinguishability of Random Sponges",
	booktitle="Advances in Cryptology -- CRYPTO 2019",
	year="2019",
	publisher="Springer International Publishing",
	address="Cham",
	pages="296--325"}

@misc{MMW24,
	author = {Christian Majenz and Giulio Malavolta and Michael Walter},
	title = {Permutation Superposition Oracles for Quantum Query Lower Bounds},
	howpublished = {Cryptology {ePrint} Archive, Paper 2024/1140},
	year = {2024},
	url = {https://eprint.iacr.org/2024/1140}
}

@inproceedings{CoudronM19,
    author =	 {Matthew Coudron and Sanketh Menda},
    title =	 {Computations with Greater Quantum Depth Are Strictly More Powerful (Relative to an Oracle)},
    booktitle =	 { Proceedings of the 52nd Annual ACM Symposium on Theory of Computing},
    pages =	 {889–901},
    year =	 2020,
    doi =		 {10.1145/3357713.3384269},
    eprint={1909.10503},
      archivePrefix={arXiv}
}

@inproceedings{BCGKPW20,
	author = {Shalev Ben-David and Andrew M. Childs and Andr{\'a}s Gily{\'e}n and William Kretschmer and Supartha Podder and Daochen Wang},
	booktitle = {Proceedings of the 61st IEEE Symposium on Foundations of Computer Science},
	doi = {10.1109/FOCS46700.2020.00066},
	eprint = {2006.12760},
    archivePrefix={arXiv},
	pages = {649-660},
	title = {Symmetries, graph properties, and quantum speedups},
	year = {2020}
}

@inproceedings{GV20,
  author={András Gilyén and Umesh Vazirani},
  title={({Sub})Exponential advantage of adiabatic quantum computation with no sign problem},
  eprint={2011.09495},
  archivePrefix={arXiv},
  booktitle={Proceedings of STOC 2021},
  pages={1357-1369},
  year={2021}}

@misc{LWWZ25,
  author={Jiaqi Leng and Kewen Wu and Xiaodi Wu and Yufan Zheng},
  title={({Sub})Exponential Quantum Speedup for Optimization},
  eprint={2504.14841},
  archivePrefix={arXiv},
  year={2025}
}

@misc{GL24,
  author={Allan Grønlund and Kasper Green Larsen},
  title={An Exponential Separation Between Quantum and Quantum-Inspired Classical Algorithms for Machine Learning},
  eprint={2411.02087},
  archivePrefix={arXiv},
  year={2024}}

@article{Llo96,
	author = {Seth Lloyd},
	doi = {10.1126/science.273.5278.1073},
	journal = {Science},
	number = {5278},
	pages = {1073-1078},
	title = {Universal quantum simulators},
	volume = {273},
	year = {1996}}

@article{Aaronson21,
author = {Aaronson, Scott}, 
title = {Open Problems Related to Quantum Query Complexity}, 
year = {2021}, 
volume = {2}, 
number = {4}, 
doi = {10.1145/3488559}, 
journal = {ACM Transactions on Quantum Computing},
articleno = {14}, 
numpages = {9},
pages={1-9},
eprint={2109.06917},
archivePrefix={arXiv}}

@misc{Li23,
  author = {Jianqiang Li},
  title = {Exponential speedup of quantum algorithms for the pathfinding problem},
  eprint={2307.12492},
  archivePrefix={arXiv},
  year={2023}
}

@misc{LZ23,
  author = {Jianqiang Li and Sebastian Zur},
  title = {Multidimensional Electrical Networks and their Application to Exponential Speedups for Graph Problems},
  eprint = {2311.07372},
  archivePrefix={arXiv},
  year={2023}
}

@misc{LT24,
  author = {Jianqiang Li and Yu Tong},
  title = {Exponential Quantum Advantage for Pathfinding in Regular Sunflower Graphs},
  eprint = {2407.14398},
  archivePrefix={arXiv},
  year={2024}
}

@article{Rosmanis11,
	author = {Rosmanis, Ansis},
	doi = {10.1103/PhysRevA.83.022304},
	eprint = {1004.4054},
    archivePrefix={arXiv},
	journal = {Physical Review A},
	number = {2},
	pages = {022304},
	title = {Quantum snake walk on graphs},
	volume = {83},
	year = {2011}
}

@InProceedings{CCG22,
  author =	{Childs, Andrew M. and Coudron, Matthew and Gilani, Amin Shiraz},
  title =	{Quantum algorithms and the power of forgetting},
  booktitle =	{14th Innovations in Theoretical Computer Science Conference (ITCS 2023)},
  pages =	{37:1--37:22},
  series =	{Leibniz International Proceedings in Informatics (LIPIcs)},
  ISSN =	{1868-8969},
  year =	{2023},
  volume =	{251},
  doi =		{10.4230/LIPIcs.ITCS.2023.37},
  eprint = {2211.12447},
  archivePrefix={arXiv}
}

@misc{MP26,
  title = {Hardness of Pathfinding in a Welded Tree},
  author = {David Miloschewsky and Supartha Podder},
  eprint = {2609.20651},
  archivePrefix={arXiv},
  year = {2026}
}

@inproceedings{CGP26,
  author =	{Cornelissen, Arjan and Gilani, Amin Shiraz and Patro, Subhasree},
  title =	{Quantum Algorithms for Path and Cycle Containment Problems},
  booktitle =	{Approximation, Randomization, and Combinatorial Optimization. Algorithms and Techniques (APPROX/RANDOM 2026)},
  pages =	{72:1--72:23},
  series =	{Leibniz International Proceedings in Informatics (LIPIcs)},
  year =	{2026},
  volume =	{392},
  doi =		{10.4230/LIPIcs.APPROX/RANDOM.2026.72},
  eprint={2605.09017},
  archivePrefix={arXiv},
}

\newpage

\appendix
\section{Clipping the color-word subspace}
\label{sec:appendix}

We prove \Cref{lem:placement-lift,lem:clipped-subspace}, using the
notation of \Cref{subsec:clipping-placements}. In particular, a block
has $\sqrt N$ leaves, a color-$c$ matching has size $K_c$ with
$N/2\le K_c\le N$, and the size of a description $(A,B)$ is the
number of vertices in its core. Labels outside the blocks do not count
toward this size, and their number is not bounded in this appendix.

The proof has three steps. First we recover the core from the database
and represent the color-word states by functions of independent core
placements. Next we bound the effect of removing block collisions and
replacing independent weld samples by matching weights. Finally, we
identify the two possible clipping events and compare the original and
clipped spans. We retain the distinction between an individual
color-word state and an arbitrary vector in their span throughout.

\subsection{The core and its placement representation}
\label{app:core-representation}

The prefix forest of $(A,B)$ has a vertex for each rooted prefix of a
word in $\mathbf A$ or $\mathbf B$, and an edge for each one-letter
extension. Prefixes starting at different global roots are distinct,
even if their color strings agree. Only the specified words are marked,
with their assigned labels; the other prefix vertices are unmarked.
We call edges inside the two binary trees \emph{ordinary} edges, to
distinguish them from weld edges.

We first record two facts about ordinary paths. A nonbacktracking
ordinary walk from a global root always descends: after arriving along
a parent edge, a different color can only select a child edge. Thus the
first $n/2$ letters specify the initial block, and the first $n$ letters
lead to a leaf. Also, any color-preserving map of a properly edge-colored
tree into one of the binary trees is injective. Indeed, it is locally
injective, and a nontrivial path with equal endpoint images would give a
nonbacktracking closed walk in a tree. Consequently, an ordinary path
from a leaf to a block root has exactly $n/2$ upward steps. These facts
hold even for placements that have block collisions.

For a database $D=(P,W)$ occurring in a color-word state, consider the
graph formed by all ordinary edges inside blocks and the stored weld
edges. Retain the connected components containing a recorded label or
weld. They are trees: each newly traversed weld enters a new block,
and contracting blocks gives a forest. In each such component, take
the smallest subtree containing its recorded labels and stored weld
endpoints. We refer to this subtree as the core read from the database.
Write $P_{\mathrm{top}}(D)$ for the restriction of $P$ to vertices above
the blocks, and abbreviate $\omega(D)=\omega(P,W)$.

\begin{lemma}[Recovering the core from the database]
\label{lem:app-core-type}
For a fixed nonzero color-word state $\ket{A,B}$, the top labeling
$P_{\mathrm{top}}(D)$ and the abstract marked, edge-colored core forest
are constant over its support. The latter is exactly $C(A,B)$ from
\Cref{def:core}. Here the abstract forest forgets embedded positions,
its original rooting, and which edges are welds, but retains edge
colors and marked labels. Each component contains a mark and has at
most one unmarked leaf. If present, that leaf is a weld endpoint in
the component's initial block, and its incident edge is a weld edge.
\end{lemma}

\begin{proof}
Each new block visit enters a previously unvisited block. Thus each
stored weld attaches a new block; it cannot connect already explored
components or create a cycle. Each occupied component has one initial
block. Since each maximal specified word ends at a mark, every occupied
component contains a mark.

Words ending outside the blocks follow fixed ordinary paths, so their
label records are fixed. Group the remaining words by their rooted
prefix of length $n/2$. In the prefix tree of each group, start at this
prefix and follow the unique child until reaching a mark, a vertex with
two children, or a prefix of length $n$, whichever comes first. Let $z$
be this vertex. It is the vertex denoted $a_{p_i}$ (or its right-tree
counterpart) in \Cref{def:core}. The path preceding $z$ consists of fixed
ordinary edges and contains no mark, weld endpoint, or branch needed to
connect the terminals. If $z$ is unmarked and has only one child, its
image is a leaf of a binary tree, and the edge to its child is a weld.

Delete the vertices preceding $z$ and their incident edges. Every
remaining branch leads to a mark or a weld endpoint, so the remaining
tree is precisely the smallest subtree containing those terminals.
It is also the component prescribed by \Cref{def:core}. This
construction depends only on the words and their labels. Its leaves
are marked except possibly $z$ in the unmarked, one-child case.
In that case $z$ is the sole unmarked leaf and lies in the initial block.
Finally, incident edges of the core have distinct colors because the
realized core is a subgraph of the properly colored welded tree.
\end{proof}

\paragraph{Fixing a sector.}
A sector fixes $P_{\mathrm{top}}(D)$ and the abstract marked,
edge-colored core forest, as in \Cref{subsec:clipping-placements}.
\Cref{lem:app-core-type} shows that each color-word state has support
in exactly one sector. Different sectors have disjoint database
supports, and clipping only removes basis terms. This proves
\eqref{eq:placement-sectors} for both the original and clipped spaces.
It also proves that the size subspaces $\mathcal H_t$ are mutually
orthogonal.

Every state of core size zero is unchanged by clipping. If $x$ is
absent there is no update. Otherwise $x$ labels a vertex above the
blocks. A query that enters a block can do so because no block is
occupied, and a query staying outside the blocks causes no violation.
The missing blue edge at a global root also leaves the description
unchanged. This proves the clipping lemma for $t=0$.

For a nonempty sector $\sigma$, let its cores be $C_1,\ldots,C_q$ and
put $t_i=|C_i|$, with $\sum_i t_i\le t$. Write
$\mathcal H=\mathcal H_\sigma$ and let $\mathcal H^{x,c}$ denote the
span of the clipped states in this sector. We suppress the sector
subscript on $G_\sigma,S_{0,\sigma},S_\sigma,M_\sigma$, and
$\mathcal K_\sigma$ when working in this fixed sector. Thus $M$ below
is the placement operator $M_\sigma$.

For brevity, call a core \emph{free} if all its leaves are marked;
a singleton marked core is also free. Otherwise call it \emph{forced}
and write $a_i$ for its entry leaf. These names merely
distinguish the two cases in \eqref{eq:placement-law}: any vertex of
a free core may be an anchor, whereas only $a_i$ may anchor a forced
core. An anchor fixes the embedded position of that vertex, at a value
of positive probability.

To distinguish abstract vertices, choose the least marked label in
each core and name each vertex by the color sequence along the path
from that mark. Proper edge coloring makes these names unique; they
do not add label entries to the database.

\paragraph{The placement law.}
Recall that $\Omega$ is the set of embedded vertices in blocks,
including block roots, and that $Q_c$ is given by
\eqref{eq:placement-kernel}. For a core placement
$\phi_i:C_i\to\Omega$, \eqref{eq:placement-law} reads
\[
 \mu_i(\phi_i)=\frac{\theta_i(\phi_i)}{Z_i}
     \prod_{ab\in E(C_i)}Q_{c(ab)}(\phi_i(a),\phi_i(b)),
 \qquad X_{i,a}=\phi_i(a).
\]
For a forced core, $\theta_i$ requires $\phi_i(a_i)$ to be a leaf
eligible for the color of its incident edge. That edge must then be a
weld. For a free core, $\theta_i=1$. This law imposes neither
injectivity nor the absence of block collisions.

Conditional on an allowed anchor $X_{i,a}=v$, this distribution has
the following sampling interpretation. Root the abstract tree at $a$
and explore from $v$. Ordinary steps are deterministic. At each weld
step, choose its partner independently and uniformly from the $K_c$
eligible leaves on the opposite side, with replacement. Stop if a
requested edge would leave a block through its root's parent edge.
Call completion without such an exit \emph{survival}. The product of
the transition probabilities, conditioned on survival, is exactly
$\mu_i(\,\cdot\mid X_{i,a}=v)$. The restriction to $\Omega$ excludes
exits through block roots; it does not exclude block collisions or
repeated weld endpoints. The product law $\mu=\bigotimes_i\mu_i$
places different cores independently.

\subsubsection{Decomposing functions of anchors}

All norms in the placement model use $L^2(\mu)$, or the indicated
factor $L^2(\mu_i)$. Explicitly,
\[
 \langle f,g\rangle_{L^2(\mu_i)}
   =\sum_{\phi_i}\mu_i(\phi_i)\overline{f(\phi_i)}g(\phi_i).
\]
The notation $L^2(X_{i,a})$ means the subspace of functions of that
coordinate. As in \eqref{eq:placement-space}, let $I_i=C_i$ for a
free core and $I_i=\{a_i\}$ for a forced core. Then
\[
 \mathcal K_i=\sum_{a\in I_i}L^2(X_{i,a}),
 \qquad \mathcal K=\bigotimes_{i=1}^q\mathcal K_i.
\]
Tensor factors always retain the order $1,\ldots,q$.

The sum map for core $i$ is
\begin{equation}
 F_i:\bigoplus_{a\in I_i}L^2(X_{i,a})\longrightarrow\mathcal K_i,
 \qquad F_i((f_a)_{a\in I_i})=\sum_{a\in I_i}f_a.
\end{equation}
Its domain is an external direct sum: a tuple $(f_a)$ has squared norm
$\sum_a\norm{f_a}^2$, even when the corresponding function spaces overlap
inside $L^2(\mu_i)$.

\begin{lemma}[Decomposition by core anchors]
\label{lem:app-decomposition}
For each core $i$, there are mutually orthogonal spaces
$\mathcal V_{i,a}\subseteq L^2(X_{i,a})$, $a\in I_i$, such that
\begin{equation}
 \mathcal K_i=\bigoplus_{a\in I_i}\mathcal V_{i,a}.
\end{equation}
Consequently $F_i$ has an isometric right inverse. In particular, every
$f\in\mathcal K_i$ has a preimage satisfying
\begin{equation}
 F_i\alpha=f,\qquad \norm{\alpha}\le\norm{f}.
 \label{eq:app-anchor-right-inverse}
\end{equation}
For two distinct cores, $F_i\otimes F_j$ likewise has an isometric right
inverse.
\end{lemma}

\begin{proof}
For a forced core, $F_i$ is the identity. Suppose that $C_i$ is free, and
root the abstract tree $C_i$ at a vertex $\rho$. For $a\ne\rho$, write
$p(a)$ for its parent and $V_a$ for its descendant subtree, including $a$.

Regard $Q_{c(ap(a))}$ as a bipartite matrix, with separate copies of
$\Omega$ indexing its rows and columns. Its nonzero components are singleton
entries from ordinary edges and two constant rectangular blocks from weld
steps, one for each direction across the weld. These components have
disjoint rows and columns: a leaf eligible for a color has no ordinary
edge of that color. Let $\gamma_a$ identify the component containing
$(X_{i,a},X_{i,p(a)})$. It is a function of either endpoint alone.

Fix the variables outside $V_a$ and sum the density in \eqref{eq:placement-law} over
$V_a\setminus\{a\}$. Within any component of $Q_{c(ap(a))}$, its dependence
on the two endpoints factors. The conditional law of $X_{i,a}$ therefore
depends on the outside variables only through $\gamma_a$. For any
$g\in L^2(X_{i,a})$,
\begin{equation}
 \mathbb E[g\mid (X_{i,v})_{v\notin V_a}]
       =\mathbb E[g\mid \gamma_a].
 \label{eq:app-anchor-conditioning}
\end{equation}
Set
\begin{equation}
 \mathcal V_{i,\rho}=L^2(X_{i,\rho}),\qquad
 \mathcal V_{i,a}=\{g\in L^2(X_{i,a}):\mathbb E[g\mid \gamma_a]=0\}
 \quad(a\ne\rho).
\end{equation}
For distinct vertices $a,b$, one of them, say $a$, is not an ancestor of
the other. Then $b\notin V_a$, and \eqref{eq:app-anchor-conditioning} shows
that every vector in $\mathcal V_{i,a}$ is orthogonal to every vector in
$\mathcal V_{i,b}$.

These spaces span $\mathcal K_i$. Indeed, for $a\ne\rho$ and
$g\in L^2(X_{i,a})$, write
\begin{equation}
 g=(g-\mathbb E[g\mid \gamma_a])+\mathbb E[g\mid \gamma_a].
\end{equation}
The first term lies in $\mathcal V_{i,a}$, and the second is a function of the
parent position $X_{i,p(a)}$. Repeating toward $\rho$ gives the claimed
decomposition. Thus
$R_i f=(\Pi_{\mathcal V_{i,a}}f)_{a\in I_i}$ satisfies $F_iR_i f=f$ and
$\norm{R_i f}^2=\sum_a\norm{\Pi_{\mathcal V_{i,a}}f}^2=\norm{f}^2$.
Finally, $R_i\otimes R_j$ is an isometric right inverse of
$F_i\otimes F_j$.
\end{proof}

For each allowed anchor, use the normalized indicator
\begin{equation}
 b_{i,a,v}\coloneqq\frac{\ind_{\{X_{i,a}=v\}}}
                       {\sqrt{\Pr_{\mu_i}[X_{i,a}=v]}}.
\end{equation}
For fixed $a$, these indicators form an orthonormal basis of
$L^2(X_{i,a})$. Restricting $F_i$ to a collection of anchors means
restricting its domain to the span of their indicators in the corresponding
summands. The following estimate converts conditional probabilities into
norm bounds for these restrictions.

\begin{lemma}[Restriction to an event]
\label{lem:app-restriction}
Let $F_{i,*}$ be the restriction of $F_i$ to a collection of allowed
anchors involving $\ell$ distinct core vertices. If
$\Pr_{\mu_i}[B\mid X_{i,a}=v]\le p$ for every anchor in this collection,
then
\begin{equation}
 \norm{\ind_BF_{i,*}}^2\le\ell p.
 \label{eq:app-event-restriction}
\end{equation}
The same bound holds for a restriction of $F_i\otimes F_j$ to a collection
of pairs of anchors in distinct cores: $p$ bounds the conditional
probability given both anchors, and $\ell$ counts the distinct pairs of core
vertices.
\end{lemma}

\begin{proof}
Within a fixed summand, the anchor indicators have disjoint supports.
Hence every function $f_a$ using only the selected anchor values satisfies
$\norm{\ind_B f_a}^2\le p\norm{f_a}^2$. Cauchy--Schwarz over the $\ell$
summands gives
\begin{equation}
 \norm{\ind_B\sum_a f_a}^2
 \le \ell\sum_a\norm{\ind_B f_a}^2
 \le\ell p\sum_a\norm{f_a}^2.
\end{equation}
For two independent cores, the product indicators are normalized and have
disjoint supports within each pair of summands. The same argument applies.
\end{proof}

\subsubsection{From placements to color-word states}

For a placement of positive $\mu$-probability, call an abstract edge
ordinary or a weld edge according to its image in the underlying graph.
Its \emph{ordinary components} are those left after deleting the weld
edges. A placement is \emph{good} if distinct ordinary components,
within or between cores, occupy distinct blocks. Let
$\Phi_{\mathrm{good}}$ be the set of good placements in
$\Phi=\prod_{i=1}^q\{\phi_i:C_i\to\Omega\}$, and let $G$ be multiplication
by its indicator in $L^2(\mu)$.

For $\phi\in\Phi_{\mathrm{good}}$, define $D_\phi=(P_\phi,W_\phi)$ by recording the fixed
top labels, the labels at the images of the core marks, and the images of
all weld edges of the cores. This is a database of partial injections.
Indeed, each ordinary component maps injectively by the ordinary-path observation above,
and different components occupy different blocks. Thus all core vertices
have distinct images. Proper edge coloring also ensures that two distinct
weld edges of the same color have no common endpoint. If $e_c(\phi)$
counts the abstract color-$c$ edges mapped to weld edges, then
$|W_{\phi,c}|=e_c(\phi)\le K_c$.

On good placements, the independent placement measure in \eqref{eq:placement-law}
has density
\begin{equation}
 \mu(\phi)=\left(\prod_{i=1}^q Z_i^{-1}\right)
                 \prod_{c\in\{r,b,g\}}K_c^{-e_c(\phi)}.
\end{equation}
To replace these independent-sampling weights by the matching weights in
\eqref{eq:placement-weight}, multiply their square roots by
\begin{equation}
 S_0(\phi)\coloneqq
 \prod_{c\in\{r,b,g\}}
       \left(\frac{K_c^{e_c(\phi)}}{(K_c)_{e_c(\phi)}}\right)^{1/2}.
 \label{eq:app-matching-correction}
\end{equation}
The labeling factor in \eqref{eq:placement-weight} is constant: $|P_\phi|$ is the
fixed number of top labels plus the fixed number of core marks. In
particular, this calculation does not require $|P_\phi|\le t$.
The good set is nonempty, since the sector contains a database realized by
a color-word state. Define
\begin{equation}
 S(\phi)\coloneqq
 \begin{cases}
  \dfrac{S_0(\phi)}{\max_{\psi\in\Phi_{\mathrm{good}}}S_0(\psi)},
       &\phi\in\Phi_{\mathrm{good}},\\[1ex]
  1,&\text{otherwise}
 \end{cases}
 \label{eq:app-normalized-correction}
\end{equation}
and
\begin{equation}
 M\coloneqq G S.
 \label{eq:app-placement-operator}
\end{equation}
Thus $M$ is a contraction on $L^2(\mu)$, and it commutes with every event
indicator.

\begin{lemma}[Exact color-word state representation]
\label{lem:app-representation}\leavevmode
\begin{enumerate}[nosep,label=(\roman*)]
\item The map $\phi\mapsto D_\phi$ is a bijection from
      $\Phi_{\mathrm{good}}$ onto the databases of sector $\sigma$.
\item Choose one allowed anchor $X_{i,a}=v$ in each core, and let $b$ be
      the product of their indicators. Then
\begin{equation}
 \sum_{\phi\in\Phi_{\mathrm{good}}}
       \sqrt{\mu(\phi)}\,(Mb)(\phi)\ket{D_\phi}
 \label{eq:app-anchor-state}
\end{equation}
is a scalar multiple of a color-word state of $\sigma$, or zero. Conversely, every
color-word state of $\sigma$ is obtained, up to a nonzero scalar, from such a choice
of anchors.
\end{enumerate}
\end{lemma}

\begin{proof}
For (i), contracting the ordinary components of each $C_i$ gives a tree,
and distinct components occupy distinct blocks. For a free core every
leaf is marked, so every edge is needed to connect its marks. For a forced
core the unique unmarked leaf is a weld endpoint, and every other leaf is
marked. Thus the cores read from $D_\phi$ are exactly $C_1,\ldots,C_q$.

To realize $D_\phi$ by a description, choose an ordinary component in each
free core, or the component containing $a_i$ in each forced core. Let $z$
be the vertex of this component closest to its block root. Enter from the
global root along the ordinary path to $z$, then follow the core paths to
all marks. The entry path meets the core first at $z$. Different cores
use disjoint blocks, and the contracted component trees prevent a return
to a previously visited block. Adding the fixed top labels therefore gives
a description without block violations with database $D_\phi$, so $D_\phi\in\sigma$.

Conversely, every database of $\sigma$ gives a good placement of its cores.
Each occupied component is a tree, and its smallest subtree containing
the recorded labels and weld endpoints is unique. The marked labels and
colored paths recover every formally named vertex. This recovers $\phi$
uniquely and proves the bijection.

For (ii), an anchor $X_{i,a}=v$ in a free core determines its entire
ordinary component: trace colored edges from $a$ and stop each branch when
a weld is required. Its vertex $z$ closest to the block root is therefore
fixed by the anchor. Use the description just constructed, entering at
$z$ and following the core paths to all marks. Every leaf is marked, so
all these core edges and endpoints are necessary. Its good placements are
exactly those with $X_{i,a}=v$. In a forced core, enter along the ordinary
path to the anchored leaf $a_i$, then follow the core paths to its marks.
The incident weld is required to reach every mark, so this description
also has exactly the specified core and anchor constraint.

Combining these descriptions and the fixed top labels gives a color-word state
whenever the restricted good set is nonempty. If the initial blocks
coincide, that set is empty. Conversely, the construction in
\Cref{lem:app-core-type} fixes one allowed anchor per core of every color-word state.
Thus products of anchor indicators give precisely the color-word state supports.
The states $\ket{A,B}$ are normalized, but their weights are proportional
to $\omega$; a nonzero scalar per state does not change any of the spans.
By \Cref{eq:placement-weight,eq:app-matching-correction,eq:app-normalized-correction,eq:app-placement-operator}, the coefficient
$\sqrt{\mu(\phi)}S(\phi)$ is a constant multiple of
$\sqrt{\omega(D_\phi)}$ on the good set. This proves the assertion about
weights as well as supports.
\end{proof}

By part (i), identifying a function $f\in L^2(\mu)$ supported on good
placements with
\begin{equation}
 \sum_{\phi\in\Phi_{\mathrm{good}}}
                   \sqrt{\mu(\phi)}f(\phi)\ket{D_\phi}
 \label{eq:app-weighted-identification}
\end{equation}
is an isometry onto the database space of $\sigma$. Products of allowed
anchor indicators span $\mathcal K$, so part (ii) gives the identity
\begin{equation}
 \mathcal H=M\mathcal K.
 \label{eq:app-exact-representation}
\end{equation}
We use this identification throughout the rest of the proof. 

\subsection{Collision estimates}
\label{app:placement-lifts}

We now show that removing block collisions and correcting the weights
preserves the norm up to a constant factor, after excluding functions
that are necessarily annihilated. We first bound collision probabilities
conditional on anchors, then use \Cref{lem:app-restriction} to obtain norm
bounds. We obtain a preimage in $\mathcal K$ of norm at most twice the norm
of the corresponding color-word state vector.

\subsubsection{Independent exploration estimates}

For the rest of the proof, assume $t\ge1$ and
\begin{equation}
 \tau\coloneqq\frac{t^2}{N^{1/4}}\le\frac1{64}.
 \label{eq:app-small-regime}
\end{equation}
The zero-core case was handled above; if $t\ge1$ and
\eqref{eq:app-small-regime} fails, the right side of \eqref{eq:clipping-bound} is $1$ and the
claim is automatic. Since $t\ge1$ and $\sqrt N\ge1$,
\begin{equation}
 \frac{t}{\sqrt N}=\frac{\tau}{tN^{1/4}}\le\tau\le\frac1{64}.
 \label{eq:app-small-ratio}
\end{equation}
The initial ordinary component cannot overflow: its positions are fixed
by the anchor, so an exit there would make the anchor probability zero.

Any later overflow must be preceded by a weld edge
and then exactly $n/2$ upward ordinary steps. For each possible exit edge of
the abstract tree, the location of that weld on the path from the anchor is
fixed. In each destination block, at most one arrival leaf has the required
upward color sequence: its position is determined by tracing the sequence
backward from the block root. The chance of such an arrival is at most
$\sqrt N/K_c\le2/\sqrt N$. There are at most $t_i$ possible exit edges.
Write $\Pr_{\mathrm{exp}}$ for probabilities under the independent
exploration before conditioning on survival, with the indicated
anchors fixed. For every allowed anchor,
\begin{equation}
 \Pr_{\mathrm{exp}}[\text{survival}\mid X_{i,a}=v]
                         \ge1-2t_i/\sqrt N.
 \label{eq:app-survival}
\end{equation}

Let $\mathsf I_i$ be the event that distinct ordinary components of core $i$ occupy
the same block, and $J_{ij}$ the event that cores $i,j$ have a common block.
Recall that $\mu_i$, defined in \eqref{eq:placement-law}, is the distribution on
placements of core $i$. Conditional on an allowed anchor, it is the law
of the independent exploration described above, conditioned on survival.
For two cores, the product measure $\mu_i\otimes\mu_j$ describes independent
placements, each with this conditioning. Block collisions are still allowed.

\begin{lemma}[Uniform exploration estimates]
\label{app:exploration-estimates}
Under \eqref{eq:app-small-regime}, the following bounds hold for allowed anchors.
\begin{enumerate}[nosep]
\item For any anchor $X_{i,a}=v$,
\begin{equation}
 \Pr_{\mu_i}[\mathsf I_i\mid X_{i,a}=v]\le\frac{8t_i^2}{\sqrt N}.
\end{equation}

\item For anchors $X_{i,a}=v$ and $X_{j,b}=w$ in distinct cores,
      with $v,w$ in different blocks,
\begin{equation}
 \Pr_{\mu_i\otimes\mu_j}
 [J_{ij}\mid X_{i,a}=v,\ X_{j,b}=w]
 \le\frac{8t_it_j}{\sqrt N}.
\end{equation}

\item For an anchor $X_{i,a}=v$ outside a specified block $B_u$,
\begin{equation}
 \Pr_{\mu_i}[\text{core }i\text{ visits }B_u\mid X_{i,a}=v]
 \le\frac{8t_i}{\sqrt N}.
\end{equation}

\item Suppose that tracing the abstract path from $a$ to $r$,
      starting at $v$, requires a weld step before reaching $r$.
      Then
\begin{equation}
 \Pr_{\mu_i}[X_{i,r}\text{ is a block root}\mid X_{i,a}=v]
 \le\frac8{\sqrt N}.
\end{equation}
\end{enumerate}
\end{lemma}

\begin{proof}
First consider the independent explorations with the stated anchors fixed, before conditioning on survival.
A weld sample hits any specified
block with probability at most $2/\sqrt N$, since a block has $\sqrt N$ leaves and $K_c\ge N/2$. At most $t_i$ blocks
are seen and at most $t_i$ welds are sampled. This gives $2t_i^2/\sqrt N$ for an
internal collision and $2t_i/\sqrt N$ for visiting a specified block outside the
initial one.

For two distinct initial blocks, first explore core $i$. Its chance to hit
the initial block of $j$ is at most $2t_i/\sqrt N$. Then explore $j$; its chance
to hit any of the at most $t_i$ blocks of $i$ is at most $2t_it_j/\sqrt N$. Their
sum is at most $4t_it_j/\sqrt N$.

For the last assertion, if $r$ maps to a block root, the ordinary path
from the last weld arrival to $r$ consists of exactly $n/2$ upward steps, by the ordinary-path observation above.
Thus the last weld position on the
abstract anchor-to-$r$ path is fixed. At most one arrival leaf in each block can produce that
path, giving probability at most $2/\sqrt N$.

To pass to the placement law, divide the one-core estimates by
$\Pr_{\mathrm{exp}}[\text{survival}\mid X_{i,a}=v]
\ge1-2t_i/\sqrt N$. For two cores, divide by
\[
 \Pr_{\mathrm{exp}}[\text{both explorations survive}
       \mid X_{i,a}=v,\ X_{j,b}=w]
 \ge(1-2t_i/\sqrt N)(1-2t_j/\sqrt N).
\]
The product bound holds because the explorations are independent with
the two anchors fixed. By \eqref{eq:app-small-ratio}, the denominators
are at least $31/32$ and $(31/32)^2$, respectively. The constant eight
therefore covers all four estimates under $\mu_i$ or
$\mu_i\otimes\mu_j$. No conditioning on the absence of collisions is
used in any of these estimates.
\end{proof}

\subsubsection{Collisions within and between cores}

For $1\le i<j\le q$, split the product indicators $b_{i,a,v}\otimes
b_{j,b,w}$ according to whether $v$ and $w$ lie in the same block.
Let $\mathcal B_{ij}\subseteq\mathcal K_i\otimes\mathcal K_j$ be the span of the products
whose anchors lie in the same block. Every vector in $\mathcal B_{ij}$ is
supported on $J_{ij}$.

\begin{lemma}[Pairwise collision estimate]
\label{lem:app-pair}
Let $1 \le i<j \le q$. Then
\begin{equation}
 \norm{\ind_{J_{ij}}g}^2
       \le\frac{8t_i^2t_j^2}{\sqrt N}\norm{g}^2
 \quad
 (g\in(\mathcal K_i\otimes\mathcal K_j)\cap\mathcal B_{ij}^\perp).
 \label{eq:app-pair-bound}
\end{equation}
Also $\norm{\ind_{\mathsf I_i}\Pi_{\mathcal K_i}}^2\le8t_i^3/\sqrt N$.
\end{lemma}

\begin{proof}
By \Cref{lem:app-decomposition}, $F_i\otimes F_j$ has an isometric right
inverse.
Split its domain according to whether the two anchor positions lie in the
same block, and denote the restrictions by $F_{\rm same}$ and $F_{\rm
different}$. Choose a preimage of $g$ with norm at most $\norm{g}$. Since
$\mathcal B_{ij}=\operatorname{range} F_{\rm same}$ and $g\perp\mathcal B_{ij}$, projecting off
$\mathcal B_{ij}$ gives
\begin{equation}
 g=(I-\Pi_{\mathcal B_{ij}})F_{\rm different}\alpha,
 \qquad \norm{\alpha}\le\norm{g}.
\end{equation}
Multiplication by $\ind_{J_{ij}}$ commutes with $\Pi_{\mathcal B_{ij}}$, because
$\mathcal B_{ij}$ is supported on that event. For different-block anchors,
\Cref{app:exploration-estimates}(2) bounds the conditional collision probability by
$8t_it_j/\sqrt N$. Apply \Cref{lem:app-restriction} to the at most $t_it_j$ pairs of
allowed variables to obtain
\begin{equation}
 \norm{\ind_{J_{ij}}F_{\rm different}}^2
                          \le8t_i^2t_j^2/\sqrt N,
\end{equation}
which proves \eqref{eq:app-pair-bound}. For an internal collision, apply
\Cref{lem:app-restriction} with at most $t_i$ variables and conditional
probability $8t_i^2/\sqrt N$ from \Cref{app:exploration-estimates}(1), then use a preimage
under $F_i$ of norm at most the norm of the vector.
\end{proof}

\subsubsection{A norm bound after removing collisions}

Collect the products that force a collision into the subspace
\begin{equation}
 \mathcal B\coloneqq\sum_{1\le i<j\le q}\mathcal B_{ij}\otimes
                         \bigotimes_{\substack{1\le k\le q\\k\notin\{i,j\}}}\mathcal K_k,
 \qquad \mathcal K_0\coloneqq\mathcal K\cap\mathcal B^\perp.
 \label{eq:app-collision-kernel}
\end{equation}
If $q=1$, the sum is empty and $\mathcal B=\{0\}$. Every vector in $\mathcal B$ is
annihilated by $G$, and hence by $M$.

\begin{lemma}[Norm bound and bounded preimages]
\label{lem:app-preimage}
For $f\in\mathcal K_0$,
\begin{equation}
 \norm{(I-G)f}^2\le\frac{16t^4}{\sqrt N}\norm{f}^2,
 \qquad \norm{Mf}\ge\tfrac12\norm{f}.
 \label{eq:app-preimage}
\end{equation}
Moreover, every $h\in\mathcal H$ has a preimage $f\in\mathcal K_0$ satisfying
\begin{equation}
 Mf=h,\qquad \norm{f}\le2\norm{h}.
 \label{eq:app-bounded-lift}
\end{equation}
\end{lemma}

\begin{proof}
The event excluded by $G$ is the union of the events $\mathsf I_i$ and $J_{ij}$.
For each pair, $f\perp\mathcal B$ places it in $\mathcal B_{ij}^\perp$ on those two
factors, with the other factors unrestricted in their $\mathcal K_k$. Thus
\Cref{lem:app-pair} applies after tensoring with the identity on the remaining
factors, without an additional factor in the bound. Using the union bound
pointwise,
\begin{align}
 \norm{(I-G)f}^2
 &\le\left(\frac8{\sqrt N}\sum_{i=1}^q t_i^3
          +\frac8{\sqrt N}\sum_{1\le i<j\le q}t_i^2t_j^2\right)\norm{f}^2\\
 &\le\frac{16t^4}{\sqrt N}\norm{f}^2.
 \label{eq:app-collision-bound}
\end{align}

The matching correction is also uniform. Every recorded weld is a core edge,
so $\sum_{c\in\{r,b,g\}} e_c\le\sum_{i=1}^q(t_i-1)\le t$. Since $K_c\ge
N/2$ and $t/\sqrt N\le1/64$, we have $t/K_c\le 2t/N\le1/32$. Using
$-\log(1-z)\le2z$ for $0\le z\le1/2$ gives
\begin{equation}
 0\le\log S_0
 =\frac12\sum_{c\in\{r,b,g\}}\sum_{j=0}^{e_c-1}-\log(1-j/K_c)
 \le\sum_{c\in\{r,b,g\}}\sum_{j=0}^{e_c-1}\frac{j}{K_c}
 \le\frac{t^2}{N}\le\frac1{64^2}.
\end{equation}
In particular, on good placements,
\[
 e^{-t^2/N}\le S\le1,
 \qquad \norm{I-S}\le1-e^{-t^2/N}\le t^2/N.
\]
Together with the collision estimate, this also gives, for
$f\in\mathcal K_0$,
\begin{align}
 \norm{(I-M)f}
 &\le\norm{(I-G)f}+\norm{G(I-S)f}\\
 &\le\bigl(4t^2N^{-1/4}+t^2/N\bigr)\norm{f}
   =\eta_t\norm{f}.
 \label{eq:app-placement-error}
\end{align}
Also $e^{-1/64}\le S\le1$ on good placements and
$16t^4/\sqrt N=16\tau^2\le1/256$. These bounds give
\begin{equation}
 \norm{Mf}\ge e^{-1/64}\sqrt{1-1/256}\norm{f}\ge\norm{f}/2.
\end{equation}

Finally, $\mathcal K=\mathcal B\oplus\mathcal K_0$ and $M\mathcal B=0$, so
$M\mathcal K_0=M\mathcal K=\mathcal H$ by \eqref{eq:app-exact-representation}. The lower bound in
\eqref{eq:app-preimage} gives \eqref{eq:app-bounded-lift}.
\end{proof}

Taking $\mathcal K_{\sigma,0}=\mathcal K_0$ proves the remaining
assertions of \Cref{lem:placement-lift}, including
\eqref{eq:placement-lift}. The anchor decomposition was proved above,
and the zero-core representation is one-dimensional with $G=S=M=I$.
All these estimates are uniform in the fixed top labels. They remain
valid after tensoring with an identity operator on any untouched
register: expanding in an orthonormal basis of that register and
summing the squared estimates proves the same bounds.

\subsection{The two clipping events}\label{app:clipping-events}

We finish by identifying the two ways a query can invalidate an
exploration: entry into an occupied block and overflow from a block.
For each, we separate color-word states that remain unchanged from those that lose
the offending event, and bound that event on the orthogonal complement
of the unchanged span. Together with the norm bound from
\Cref{app:placement-lifts}, these estimates give a direct comparison of the
original and clipped color-word state spans.

Fix a sector and the query $(x,c)$. The sector determines whether $x$
is absent, labels a fixed vertex outside the blocks, or marks a fixed
abstract vertex of a core. 
If $x$ is absent or the query stays outside the blocks, clipping
has no effect; backtracking and absent global-root edges are treated as in
\Cref{subsec:clipping-placements}.
Otherwise exactly one of the two cases below
applies, and which one is determined by the sector.
In the regime indicated by
\eqref{eq:app-small-regime} these are the only restrictions, as verified after the two cases are stated.
We use $E$ and its subscripted variants
for multiplication by the indicators of the specified events.

\begin{itemize}
\item \textbf{Entering a block from above.}
The label $x$ is at a fixed vertex of depth $n/2-1$, and color $c$ leads to
the root of a fixed block $B_u$. A color-word state with $B_u$ as an initial block is
unchanged. Every other color-word state is restricted to databases in which $B_u$ is
unoccupied. Write $E_u$ for multiplication by the indicator that at least
one core occupies $B_u$.

\item \textbf{Overflow from a core.}
The label $x$ marks a vertex $r$ of some $C_j$. Let $E$ be multiplication by
the indicator that $X_{j,r}$ is a block root and $c$ is the color of its
ordinary parent edge leaving the block through a non-weld edge. On this event, if $r$ is in the
color-word state's initial block, the query returns to an existing rooted prefix and
the color-word state is unchanged. Otherwise the color-word state is restricted by $I-E$.
\end{itemize}

To see that these are the only restrictions, first consider an ordinary step
inside a block. It cannot reach a different ordinary component, since those
components occupy different blocks. If it reaches a vertex already explored
in the same component, the unique ordinary path to that vertex already
occurs in the prefix forest. Indeed, the explored part of a component is a
connected subtree of its block, so if the neighbor reached is already explored
then the traversed edge itself lies in that subtree; the step therefore either
follows the existing prefix of that neighbor or deletes the last letter of $w$,
and no new rooted prefix arises.

Next consider an unrecorded weld edge. There are at
most $t$ occupied blocks, because each contains a core vertex. On the
destination side, at least $K_c-t\sqrt N\ge N(1/2-t/\sqrt N)>0$ eligible leaves lie
outside them, by \eqref{eq:app-small-ratio}. Choosing any such leaf respects the
partial matching and enters a new block. A fresh label is available as
in the extension convention of \Cref{subsec:clipping-placements}:
fewer than $|V|=4N-2<N^2$ vertices are already labeled whenever a new
vertex must be marked. This remains true regardless of the number of
top labels.

Finally, a
recorded weld edge already belongs to the prefix forest and imposes no new
restriction.

\subsubsection{Entering an occupied block}

Fix $B_u$ and, in each factor, let $E_{i,u}$ be multiplication by the
indicator that core $i$ visits $B_u$. Let $\mathcal J_{i,u}\subseteq\mathcal K_i$ be the
span of the allowed indicators $b_{i,a,v}$ with $v\in B_u$. Then
$E_{i,u}\mathcal J_{i,u}=\mathcal J_{i,u}$.

\begin{lemma}[Entry estimate]
\label{lem:app-entry}
With
\begin{equation}
 \mathcal J=\sum_{i=1}^q\mathcal J_{i,u}\otimes
               \bigotimes_{\substack{1\le k\le q\\k\ne i}}\mathcal K_k,
\end{equation}
one has $E_u\mathcal J=\mathcal J$ and
\begin{equation}
 \norm{E_ug}^2\le\frac{8t^2}{\sqrt N}\norm{g}^2
             \qquad(g\in\mathcal K\cap\mathcal J^\perp).
 \label{eq:app-entry-bound}
\end{equation}
The clipped color-word state span $\mathcal H^{x,c}$ in this sector is
\begin{equation}
 \mathcal H^{x,c}=M\mathcal J+(I-E_u)M\mathcal K.
 \label{eq:app-entry-span}
\end{equation}
\end{lemma}

\begin{proof}
Split $F_i$ into the parts with anchors in $B_u$ and outside $B_u$, and write
$F_{i,\mathrm{out}}$ for the latter part. For
$g_i\in\mathcal K_i\cap\mathcal J_{i,u}^\perp$, take a preimage under $F_i$ as in
\eqref{eq:app-anchor-right-inverse} and project off $\mathcal J_{i,u}$. This gives
$g_i=(I-\Pi_{\mathcal J_{i,u}})F_{i,\mathrm{out}}\alpha$ with $\norm{\alpha}\le\norm{g_i}$. Apply \Cref{lem:app-restriction} with at most $t_i$ variables and the
conditional probability $8t_i/\sqrt N$ from
\Cref{app:exploration-estimates}(3) to obtain
\begin{equation}
 \norm{E_{i,u}g_i}^2\le\frac{8t_i^2}{\sqrt N}\norm{g_i}^2.
\end{equation}
Here $E_{i,u}$ commutes with $\Pi_{\mathcal J_{i,u}}$.

If $g\perp\mathcal J$, it is in $\mathcal J_{i,u}^\perp$ on each factor $i$. Tensor
the preceding estimate with the other factors. If a core visits $B_u$, at
least one of the events defining $E_{i,u}$ occurs. Hence
\begin{equation}
 \norm{E_ug}^2\le\sum_{i=1}^q\norm{E_{i,u}g}^2
      \le\frac8{\sqrt N}\sum_{i=1}^q t_i^2\norm{g}^2
      \le\frac{8t^2}{\sqrt N}\norm{g}^2.
\end{equation}

By \Cref{lem:app-representation}, $M\mathcal J$ is exactly the span of color-word states with
initial block $B_u$; call these color-word states \emph{safe}. Other color-word states are clipped
by $I-E_u$. Applying $I-E_u$ to a safe color-word state gives zero, so including all
color-word states in the second summand makes no difference. This proves
\eqref{eq:app-entry-span}.
\end{proof}

\subsubsection{Overflow through the top of a block}

Let $r\in V(C_j)$ carry $x$, and let $E$ be the overflow projector defined
above. Suppose first that $C_j$ is free. Put
\begin{equation}
 \mathcal U=L^2(X_{j,r})\subseteq\mathcal K_j.
\end{equation}
Call an anchor $X_{j,a}=v$ \emph{safe} if tracing the abstract path from $a$
to $r$, starting at $v$, reaches $r$ using only ordinary edges. This
property is determined by $a,v$. On good placements it means that $r$ lies
in the color-word state's initial block, so overflow at $r$ would merely return to an
existing prefix. A safe anchor determines $X_{j,r}$, and the reverse
ordinary path determines the anchor from that value of $X_{j,r}$. Hence its
indicator lies in $\mathcal U$. Conversely, $a=r$ supplies every indicator of a
value of $X_{j,r}$. Thus the safe anchor indicators span exactly $\mathcal U$.

\begin{lemma}[Overflow estimate]
\label{lem:app-overflow}
If $C_j$ is free, set
\begin{equation}
 \mathcal J=\mathcal U\otimes
        \bigotimes_{\substack{1\le k\le q\\k\ne j}}\mathcal K_k,
\end{equation}
with $\mathcal U$ in the $j$th tensor position. Then $E\mathcal J\subseteq\mathcal J$ and
\begin{equation}
 \norm{Eg}^2\le\frac{8t_j}{\sqrt N}\norm{g}^2
                 \qquad(g\in\mathcal K\cap\mathcal J^\perp).
 \label{eq:app-overflow-bound}
\end{equation}
If $C_j$ is forced, set $\mathcal J=\{0\}$; then
\begin{equation}
 \norm{Eg}^2\le\frac8{\sqrt N}\norm{g}^2\qquad(g\in\mathcal K).
 \label{eq:app-overflow-forced}
\end{equation}
In both cases the clipped color-word state span $\mathcal H^{x,c}$ in this sector is
\begin{equation}
 \mathcal H^{x,c}=M\mathcal J+(I-E)M\mathcal K.
 \label{eq:app-overflow-span}
\end{equation}
\end{lemma}

\begin{proof}
For a free core, restrict $F_j$ to the anchors that are not safe, and call
this map $F_{j,\mathrm{weld}}$. Taking a preimage under $F_j$ as in
\eqref{eq:app-anchor-right-inverse} and projecting off $\mathcal U$ gives, for $g_j\in\mathcal K_j\cap \mathcal U^\perp$,
\begin{equation}
 g_j=(I-\Pi_{\mathcal U})F_{j,\mathrm{weld}}\alpha,
 \qquad\norm{\alpha}\le\norm{g_j}.
\end{equation}
The indicator defining $E$ depends only on $X_{j,r}$, so $E$ commutes with
$\Pi_{\mathcal U}$. Apply \Cref{lem:app-restriction} to at most $t_j$ variables, using the
bound $8/\sqrt N$ from \Cref{app:exploration-estimates}(4), to obtain
$\norm{EF_{j,\mathrm{weld}}}^2\le8t_j/\sqrt N$. Tensoring proves
\eqref{eq:app-overflow-bound}.

For a forced core, its anchor is the entry leaf. Its ordinary
component consists of that leaf alone, so every marked vertex, including
$r$, lies outside the initial block on good placements. Every anchor-to-$r$
path starts with a weld edge. The anchor indicators have disjoint supports,
and
\Cref{app:exploration-estimates}(4) directly gives
\eqref{eq:app-overflow-forced}.

Call a color-word state \emph{safe} if its anchor in core $j$ is safe, and call
the linear span of these color-word states the \emph{safe span}. Safe color-word states are 
unchanged; all other color-word states are restricted by $I-E$. In
the free case their safe span is $M\mathcal J$. Since $E\mathcal J\subseteq\mathcal J$,
applying $I-E$ to a safe vector remains in the safe span. Thus the clipped
span has exactly the form
\eqref{eq:app-overflow-span}. In the forced case every color-word state is clipped
by $I-E$, which gives the same identity with $\mathcal J=\{0\}$.
\end{proof}

\subsubsection{Proof of the clipping lemma}

\begin{proof}[Proof of \Cref{lem:clipped-subspace}]
Sectors with no cores are unchanged by clipping, and the bound is trivial
when $t^2/N^{1/4}>1/64$. Assume \eqref{eq:app-small-regime} and fix a sector with at
least one core. All spaces and operators in the argument below refer to this 
fixed sector. The resulting bound is uniform over sectors, and we pass to the 
full spaces at the end of the proof. 

If clipping is 
nontrivial, \Cref{lem:app-entry,lem:app-overflow} give an event projector $E$ and a 
subspace $\mathcal J\subseteq\mathcal K$ such that 
\begin{equation}
 \mathcal H^{x,c}=M\mathcal J+(I-E)\mathcal H,\qquad E\mathcal J\subseteq\mathcal J,
 \qquad
 \norm{Eg}\le\alpha\norm{g}
 \quad(g\in\mathcal K\cap\mathcal J^\perp),
 \label{eq:app-clipping-data}
\end{equation}
where $\alpha=\sqrt8\,t/N^{1/4}$. Here we used $t\ge1$ and $t_j\le t$ to
cover both overflow estimates. Also $\mathcal H=M\mathcal K$ by
\eqref{eq:app-exact-representation}.

Let $\mathcal S=M\mathcal J$ be the safe span and put
$\mathcal U=\mathcal H\cap\mathcal S^\perp$. Since $M$ commutes with $E$ and
$E\mathcal J\subseteq\mathcal J$, the space $\mathcal S$ is invariant under $E$.
As $E$ is self-adjoint, it commutes with $\Pi_{\mathcal S}$.
For $u\in\mathcal U$, choose $f\in\mathcal K$ with $Mf=u$ and
$\norm{f}\le2\norm{u}$, using \eqref{eq:app-bounded-lift}. Set
$g=f-\Pi_{\mathcal J} f$. Then $g\in\mathcal K\cap\mathcal J^\perp$ and
\begin{equation}
 u=(I-\Pi_{\mathcal S})Mg,\qquad
 \norm{Eu}
 =\norm{(I-\Pi_{\mathcal S})MEg}
 \le\norm{Eg}\le2\alpha\norm{u}.
 \label{eq:app-unsafe-bound}
\end{equation}
Write $\beta=2\alpha$. Since $t\ge1$ and $t^2/N^{1/4}\le1/64$,
$\beta\le2\sqrt8/64<1/2$.

The two spaces now have the forms
\begin{equation}
 \mathcal H=\mathcal S\oplus\mathcal U,\qquad
 \mathcal H^{x,c}=\mathcal S\oplus(I-E)\mathcal U.
\end{equation}
The second sum is orthogonal because $I-E$ preserves $\mathcal S^\perp$.
For $v=a+u\in\mathcal H$, with $a\in\mathcal S$ and $u\in\mathcal U$, the vector
$a+(I-E)u$ belongs to $\mathcal H^{x,c}$, so
\begin{equation}
 \operatorname{dist}(v,\mathcal H^{x,c})\le\norm{Eu}\le\beta\norm{v}.
\end{equation}
Conversely, if $v'=a+(I-E)u\in\mathcal H^{x,c}$, then
\begin{equation}
 \norm{v'}^2=\norm{a}^2+\norm{u}^2-\norm{Eu}^2
 \ge(1-\beta^2)(\norm{a}^2+\norm{u}^2).
\end{equation}
Comparing with $a+u\in\mathcal H$ gives
\begin{equation}
 \operatorname{dist}(v',\mathcal H)
 \le\norm{Eu}
 \le\frac{\beta}{\sqrt{1-\beta^2}}\norm{v'}
 \le2\beta\norm{v'}.
\end{equation}
For two orthogonal projectors $P,Q$, their difference has norm
$\max\{\norm{(I-Q)P},\norm{(I-P)Q}\}$; this follows by squaring $P-Q$
and splitting into the range and kernel of $P$. Applying this identity to
the two distance bounds yields
\begin{equation}
 \norm{\Pi_{\mathcal H^{x,c}}-\Pi_{\mathcal H}}
 \le2\beta=4\sqrt8\,\frac{t}{N^{1/4}}
 \le64\,\frac{t^2}{N^{1/4}}.
\end{equation}

This bound is uniform over sectors. By \eqref{eq:placement-sectors}, their projector
differences act on mutually orthogonal database subspaces, so
\begin{equation}
 \norm{\Pi_{\le t}|^{x,c}-\Pi_{\le t}}
 =\max_{\sigma:\,\sum_i|C_i|\le t}
   \norm{\Pi_{\mathcal H_\sigma^{x,c}}-\Pi_{\mathcal H_\sigma}}
 \le64\,\frac{t^2}{N^{1/4}}.
\end{equation}
The difference of two orthogonal projectors always has norm at most $1$,
so this is exactly the bound in \eqref{eq:clipping-bound}.
\end{proof}

\end{document}